\documentclass[journal]{IEEEtran}

\usepackage{amsmath}
\usepackage{amsthm}
\usepackage{amsfonts}
\usepackage[tight,footnotesize]{subfigure}
\usepackage{graphicx}
\usepackage{algorithm}
\usepackage{algorithmic}
\usepackage{color}
\usepackage{amssymb,amsmath}
\usepackage{url}

\newtheorem{theorem}{Theorem}

\newtheorem{proposition}{Proposition}

\theoremstyle{definition}

\newtheorem{remark}{Remark}

\newcommand{\opt}[1]{#1^*}
\newcommand{\nonopt}[1]{#1^o}
\newcommand{\feascov}{\mathcal{P}}

\newcommand{\ccovaux}[1][1,\cdots,K]{_{#1}}

\newcommand{\ccov}[1][]{\bar{\boldsymbol{Q}}^{#1}}
\newcommand{\ccovo}[1][]{\bar{\boldsymbol{Q}}^{*#1}}
\newcommand{\ccovno}[1][]{\bar{\boldsymbol{Q}}^{o#1}}

\newcommand{\cccov}{\mathcal{Q}}
\newcommand{\cccovo}{\mathcal{Q}^*}
\newcommand{\cccovno}{\mathcal{Q}^o}

\newcommand{\ucov}[1][]{\boldsymbol{Q}^{#1}\ccovaux}
\newcommand{\ucovo}[1][]{\boldsymbol{Q}^{*#1}\ccovaux}
\newcommand{\ucovno}[1][]{\boldsymbol{Q}^{o#1}\ccovaux}

\title{Optimum Transmission Through the Multiple-Antenna Gaussian Multiple Access Channel}
\author{Daniel~Calabuig,~\IEEEmembership{Member,~IEEE}, Ramy~H.~Gohary,~\IEEEmembership{Senior~Member,~IEEE}, and~Halim~Yanikomeroglu,~\IEEEmembership{Senior~Member,~IEEE}%

\thanks{Preliminary versions of this paper were presented at the \emph{IEEE International Symposium on Information Theory} (ISIT), Istanbul, Turkey, 7-12 July 2013, and at the \emph{IEEE International Workshop on Signal Processing Advances in Wireless Communications} (SPAWC), Toronto, Canada, 22-25 June 2014.}

\thanks{The work of the first author was supported by a Marie Curie International Outgoing Fellowship (IOF) of the European Commission under the project COMIC (253990). The work of the second and third authors was supported in part by Huawei Canada Co., Ltd., and in part by the Ontario Ministry of Economic Development and Innovations ORF-RE (Ontario Research Fund---Research Excellence) program.}

\thanks{Daniel Calabuig (dacaso@iteam.upv.es) is with the Institute of Telecommunications and Multimedia Applications at the Universidad Polit\'ecnica de Valencia, Spain. Ramy H. Gohary (gohary@sce.carleton.ca) and Halim Yanikomeroglu (halim@sce.carleton.ca) are with the Department of Systems and Computer Engineering at Carleton University, Canada.}

\thanks{Copyright (c) 2015 IEEE. Personal use of this material is permitted.  However, permission to use this material for any other purposes must be obtained from the IEEE by sending a request to pubs-permissions@ieee.org.}
}

\begin{document}
\maketitle

\begin{abstract}
T{his} paper studies the optimal points in the capacity region of Gaussian multiple access channels~(GMACs) with constant fading, multiple antennas and various power constraints. The points of interest maximize general rate objectives that arise in practical communication scenarios. Achieving these points constitutes the task of jointly optimizing the time-sharing parameters, the input covariance matrices and the order of decoding used by the successive interference cancellation receiver. To approach this problem, Carath\'{e}odory's theorem is invoked to represent time-sharing and decoding orders jointly as a finite-dimensional matrix variable. This variable enables us to use variational inequalities to extend results pertaining to problems with linear rate objectives to more general, potentially nonconvex, problems, and to obtain a necessary and sufficient condition for the optimality of the transmission parameters in a wide range of problems. Using the insights gained from this condition, we develop and analyze the convergence of an algorithm for solving, otherwise daunting, GMAC-based optimization problems.
\end{abstract}


\section{Introduction}
\PARstart{I}{n} a  Gaussian multiple access channel~(GMAC) multiple users send independent signals to one destination. Such a channel model  arises in uplink communication scenarios including cellular systems  when multiple users communicate with a base station and satellite systems  when multiple ground stations communicate with a satellite~\cite{Cover}. In addition to the  relevance of the GMAC to practical communication scenarios, the analysis of this channel with  a sum-power   constraint is closely related to the analysis of the, usually less understood, Gaussian broadcast channel~(GBC). This relationship was discovered in~\cite{Goldsmith_Duality}, was further investigated in~\cite{Yu},  and was used in~\cite{BC_Capacity_Region_J} to facilitate the evaluation of the capacity region of the GBC.

The capacity region of general  multiple access channels was obtained in~\cite{Ahlswede} and~\cite{PHD_Liao}. Particularizing these results to the case of Gaussian channels, it was shown therein that corner  points on the boundary of the capacity region are achieved when  the signal of each user is Gaussian distributed with an appropriate covariance and the receiver uses successive interference cancellation~(SIC) to decode  the users' signals sequentially~\cite{Wyner}. In SIC, the receiver decodes the signal of each user while treating the  signals of the set of users interfering with it as noise. After decoding, the signal of each user  is stripped off from the signal interfereing with the signal of the remaining users. Other points on the boundary of the capacity region can be obtained by time-sharing, whereby each decoding order and collection of users' covariance matrices are used during a fraction of the signalling duration.

An alternative, yet potentially suboptimum, decoding method is the so-called successive group detection~\cite{Varanasi,Prasad}, which resembles SIC but therein, groups, rather than single users, are decoded at each decoding stage. In~\cite{Varanasi}, the performance of this technique is shown to depend, not only on the number of groups, but also on how the users are distributed across groups, and on the decoding order of these groups. In~\cite{Prasad}, the optimum ordered partition that minimizes outage probability is determined for a slow fading narrow band GMAC when perfect channel state information (CSI) is available at the receiver only, but not at the transmitters.

Achieving particular points within the GMAC capacity region was considered in~\cite{Wei_Yu_IWFA_MAC},~\cite{Fairness_Maddah_Ali} and~\cite{Mesbah_Alnuweiri}, when perfect CSI is available at the receiver and the transmitters. In particular, in~\cite{Wei_Yu_IWFA_MAC}, an iterative water-filling algorithm is considered, which, for a sum power constraint at each user, yields a rate vector on the sum-capacity facet and the input covariance matrices that achieve it. In~\cite{Fairness_Maddah_Ali}, the input
covariance matrices of all the users are optimized to maximize the sum-capacity, and those matrices are subsequently used with time-sharing to achieve a fairness criterion. In~\cite{Mesbah_Alnuweiri}, two cases are considered: the case of small  number of users, which gives rise to a scenario in which time-sharing is feasible, and the case of large number of users, which gives rise to a scenario in which time-sharing is infeasible. In~\cite{Mesbah_Alnuweiri}, fairness is not directly addressed, but therein, the weighted sum of the rates of a given subset of users is maximized while the rates of the remaining users are restricted to prescribed values. A somewhat similar philosophy was applied to other multi-user communication scenarios. For instance, in~\cite{Liu_Hou_Infocom} the GMAC-GBC duality in~\cite{Goldsmith_Duality}  is used  to develop an algorithm that maximizes  a weighted sum of the logarithms of the users' rates in the GBC. In~\cite{Maddah_Ali_Doost_Khandani_Max_Min_Fairness}, an interference channel is considered, and the decoding order that achieves the optimal max-min fairness is determined.

In this paper, we focus on the  GMAC scenario with constant fading and perfect CSI at the receiver and transmitters. In contrast with previous works, we consider the joint optimization of the set of input covariance matrices, time-sharing parameters and user orderings that maximize various rate  objectives with various classes of power constraints. To approach this problem, we invoke Carath\'{e}odory's theorem to represent time-sharing and decoding orders jointly as a finite-dimensional matrix variable. This variable enables us to use variational inequalities to extend results pertaining to problems with linear objectives to more general, potentially nonconvex, problems, and to obtain a necessary and sufficient condition for the optimality of the transmission parameters in a wide range of problems. It is shown that for a class of problems in which the power constraints are convex, it suffices for each user to use only one covariance matrix in all its allocated time slots. On the other hand, for arbitrary power constraints, if the objective is linear no time-sharing is necessary. These results significantly reduce the design complexity and, together with the necessary and sufficient condition, enable us to find solutions for, otherwise daunting, GMAC-based optimization problems.

The rest of the paper is organized as follows. In Section~\ref{model}, we present the system model and use Carath\'{e}odory's theorem to express rates as functions of the finite-dimensional matrix variable that jointly represents time-sharing and decoding orders. In Section~\ref{OptGMAC:Case02}, we study problems with linear rate objectives. In Section~\ref{NonLinear}, we use variational inequalities to extend the results obtained for linear rate objectives in Section~\ref{OptGMAC:Case02} to problems with nonlinear rate objectives. In Section~\ref{Application}, we present an algorithm to compute the optimum transmission parameters for a class of nonconvex problems. In Section~\ref{examples}, numerical examples are provided to illustrate the efficacy of the proposed algorithm in solving two instances of this class of problems. In Appendix~\ref{preliminaries}, we provide background material pertaining to variational inequalities which are used in Section~\ref{NonLinear}.

\paragraph*{Notation} The expected value of a matrix $\boldsymbol{A}$ is denoted by $\text{E}[\boldsymbol{A}]$, and the direct sum of two square matrices $\boldsymbol{B}$ and $\boldsymbol{C}$ is denoted by $\boldsymbol{B}\oplus\boldsymbol{C}$, that is, $\boldsymbol{B}\oplus\boldsymbol{C}$ is a block diagonal matrix with the matrices $\boldsymbol{B}$ and $\boldsymbol{C}$ along the diagonal. The rest of the notation is standard.

\section{System model and optimization}
\label{model}
The GMAC is composed of $K$ users transmitting to one base station. The number of transmit antennas of the $k$-th user is denoted by $N_k$, $k=1,\ldots,K$, and the number of receive antennas at the base station is denoted by $N_\text{R}$. The received signal, $\boldsymbol{y}\in\mathbb{C}^{N_\text{R}}$, is given by
\begin{equation}
\boldsymbol{y}=\sum_{k=1}^K\boldsymbol{H}_k\boldsymbol{x}_k+\boldsymbol{z},
\end{equation}
where $\boldsymbol{H}_k\in\mathbb{C}^{N_\text{R}\times N_k}$ is the channel matrix of the $k$-th user and  $\boldsymbol{x}_k\in\mathbb{C}^{N_k}$ is its transmitted  signal. The Gaussian noise at the base station is denoted by $\boldsymbol{z}\in\mathbb{C}^{N_\text{R}}$, which, without loss of generality, is assumed to satisfy $\text{E}[\boldsymbol{z}\boldsymbol{z}^\dagger]=\boldsymbol{I}$. We consider the case in which the channel matrices are known by the users and the base station.

Let $\ucov[][k]=\text{E}[\boldsymbol{x}_k\boldsymbol{x}_k^\dagger]$ be the covariance matrix of the signal of user $k$, and let $\ccov[]=\ucov[][1]\oplus\cdots\oplus\ucov[][K]\in\mathcal{D}$ be the composite covariance matrix, where $\mathcal{D}$ is the set of all block diagonal matrices in which the $k$-th block size is $N_k\times N_k$, $k=1,\dots,K$. We consider systems with $L$ power constraints which can be expressed as $g_\ell(\ccov)\leq0,\;\ell=1,\ldots,L$. Let
\begin{equation}
\feascov=\{\ccov\in\mathcal{D}\;|\;\ccov\succeq0,\;g_\ell(\ccov)\leq0,\;\ell=1,\ldots,L\}\label{eq:P_set}
\end{equation}
be the set of all feasible $\ccov$. For finite transmit powers, this set is bounded.

Studying the polymatroid structure of the GMAC capacity region, it can be shown that the corners points of the polymatroid, and hence of the capacity region, are achieved with an SIC receiver~\cite{Tse_Hanly}. To achieve a particular corner, the receiver orders the users and decodes their signals sequentially. To decode the signal of a particular user, the receiver treats the  signals of the set of users interfering with it as additive noise. After decoding, the signal of that user is stripped off from the signals of the remaining users. Since the amount of interference observed in decoding the signal of a particular user depends on the ordering, it can be seen that, to maximize a given objective, the receiver  must determine the optimal user ordering.

Let $\pi_1,\cdots,\pi_{K!}$ be the set of all $K!$ permutations of $K$ elements, where  $\pi_i(j)$ refers to  the user in the $j$-th position of the $i$-th ordering. When the receiver uses $\pi_i$ for decoding the users' signals, each user $k\in\{1,\dots,K\}$ is able to achieve the following rate~\cite{Goldsmith_Duality}:
\begin{equation}\label{eq:r_ki}
r_{ki}\left(\ccov[]\right)=\log\frac{\left|\boldsymbol{I}\!+\!\sum_{j\geq\pi_i^{-1}(k)}\boldsymbol{H}_{\pi_i(j)}\ucov[][\pi_i(j)]\boldsymbol{H}_{\pi_i(j)}^\dagger\right|}{\left|\boldsymbol{I}\!+\!\sum_{j>\pi_i^{-1}(k)}\boldsymbol{H}_{\pi_i(j)}\ucov[][\pi_i(j)]\boldsymbol{H}_{\pi_i(j)}^\dagger\right|}.
\end{equation}
Note that this rate expression corresponds to the case in which the channel coefficients are fixed and known by the users and the base station.

The GMAC capacity region is the convex hull of all the rate vectors that are achievable with all orderings and covariance matrices in $\feascov$, cf.~\eqref{eq:P_set}. This implies that, to achieve a particular rate vector in the capacity region, one may have to use convex combinations of multiple rate vectors with the corresponding composite covariance matrices and permutations. Let the number of such rate vectors be denoted by $M$, and let the corresponding composite covariance matrices and permutations be denoted by $\{\ccov[(m)]\}_{m=1}^M$ and $\{\nu_m\}_{m=1}^M$ respectively. Then, convex combinations can be implemented by using $M$ non-negative time-sharing coefficients, $a_m\geq 0$, $m=1,\ldots,M$, where $\sum_{m=1}^Ma_m=1$. These coefficients represent the percentage of time during which each of the $M$ rate vectors, with the corresponding composite covariance matrix and permutation, are used. Using this notation, with time-sharing included, the $k$-th user achieves the following rate:
\begin{equation}\label{eq:model:01}
\sum_{m=1}^Ma_m\log\frac{\left|\boldsymbol{I}\!+\!\sum_{j\geq\nu_m^{-1}(k)}\boldsymbol{H}_{\nu_m(j)}\ucov[(m)][\nu_m(j)]\boldsymbol{H}_{\nu_m(j)}^\dagger\right|}{\left|\boldsymbol{I}\!+\!\sum_{j>\nu_m^{-1}(k)}\boldsymbol{H}_{\nu_m(j)}\ucov[(m)][\nu_m(j)]\boldsymbol{H}_{\nu_m(j)}^\dagger\right|},
\end{equation}
where $\ucov[(m)][k]$ is the covariance matrix of the signal of the $k$-th user in the \mbox{$m$-th} composite covariance matrix, $\ccov[(m)]$. Since the convex hull that constitutes the GMAC capacity region is constructed using time-sharing of potentially non-connected sets, the maximum number of rate vectors required to be combined can be obtained using Carath\'eodory's theorem~\cite[Proposition~B.6]{Bertsekas}, which implies that it suffices to set $M=K+1$ for every point in the GMAC capacity region to be achieved. However, if the convex hull is constructed using one connected set, it can be shown, using the Fenchel-Eggleston strengthening of Carath\'eodory's theorem, that it suffices to set $M=K$~\cite{Utschick_Brehmer}\footnote{This reference was drawn to our attention by an anonymous reviewer.}. To maintain generality, we consider that $M=K+1$, and henceforth the index $m$ will be assumed to take on values between $1$ and $K+1$.

Being explicit in permutations, the characterization in~\eqref{eq:model:01} is not well-suited for finding  collections of composite covariance matrices, time-sharing parameters and user-orderings that maximize objective functions of practical importance. To circumvent this difficulty, time-sharing parameters and user  permutations are combined in what we refer to as the time-sharing  matrix $\boldsymbol{\alpha}\in\mathbb{R}^{(K+1)\times K!}$, with entries in the $(K+1)\times K!$ unit simplex,
\begin{equation}\label{eq:simplex_a}
\mathcal{S}_\alpha\triangleq\Bigl\{\boldsymbol{\alpha}\Bigl| \sum_{m=1}^{K+1}\sum_{i=1}^{K!}\alpha_{mi}=1,\  \alpha_{mi}\geq 0,\ \forall m,i\Bigr\}.
\end{equation}
The $mi$-th element of $\boldsymbol{\alpha}$, $\alpha_{mi}$, represents the percentage of time during which the $m$-th composite covariance matrix is used with the $i$-th ordering. In particular, $\alpha_{mi}=a_m$ if the permutation $\nu_m=\pi_i$, and $\alpha_{mi}=0$, otherwise. If $\alpha_{mi}=0$ for all $m=1,\dots,K+1$, then, the permutation $\pi_i$ is not used to achieve the rate vector. Hence, this representation renders  the elements of $\boldsymbol{\alpha}$, not only time-sharing coefficients for generating convex combinations, but also  indicators of the used orderings.
Let $\cccov=\{\ccov[(m)]\}_{m=1}^{K+1}$ be the collection of composite covariance matrices. Then, the GMAC capacity region can be expressed as  the union of  rate vectors $\boldsymbol{\rho}\left(\boldsymbol{\alpha},\cccov\right)$ with the $k$-th entry given by
\begin{equation}\label{eq:rate_k}
\rho_k\left(\boldsymbol{\alpha},\cccov\right)=\sum_{m=1}^{K+1}\sum_{i=1}^{K!}\alpha_{mi}r_{ki}\big(\ccov[(m)]\big).
\end{equation}
By introducing the matrix $\boldsymbol{\alpha}$, we obtained a form that is more convenient for subsequent optimization, but at the expense of increasing the dimensionality from combining $K+1$ rate vectors to combining $(K+1)!$ ones. However, since it suffices to time-share no more than $K+1$ rate vectors, any rate vector within the capacity region can be achieved with a time-sharing matrix with, at most, $K+1$ non zero elements. In this case, the $K+1$ rate vectors are the vectors $\{r_{ki}(\ccov[(m)])\}_{k=1}^K$ corresponding to a time-sharing coefficient $\alpha_{mi}>0$. Each of these rate vectors is achieved with the covariance matrices along the diagonal of $\ccov[(m)]$ and the decoding order in the permutation $\pi_i$.

With notation established, our goal is to develop insight into problems of the form:
\begin{subequations}\label{eq:MACP}
\begin{align}
    \underset{\boldsymbol{\alpha},\cccov}{\min}\quad & f\big(\boldsymbol{\rho}(\boldsymbol{\alpha},\cccov)\big),\label{eq:objective_1}\\
    \text{subject to}                    \quad &\boldsymbol{\alpha}\in\mathcal{S}_\alpha,\label{eq:const_1}\\
    &\ccov[(m)]\in\feascov,\; m=1,\ldots,K+1,\label{eq:const_2}
\end{align}
\end{subequations}
where throughout $f$ is assumed to be continuously differentiable, $\mathcal{S}_\alpha$ is defined in~\eqref{eq:simplex_a}, $\feascov$ is defined in~\eqref{eq:P_set}, and the time-sharing matrix $\boldsymbol{\alpha}$ is constrained to the unit simplex $\mathcal{S}_\alpha$. Carath\'eodory's theorem implies that the time-sharing matrix $\boldsymbol{\alpha}$ could be additionally constrained to have no more than $K+1$ non-zero elements. Such a constraint is nonconvex and, by virtue of Carath\'eodory's theorem, excluding it, although it constitutes an ostensible relaxation, does not enlarge the set of achievable rates. The number of real variables in this problem is $(K+1)!+(K+1)\sum_{k=1}^KN_k^2$, where the first term accounts for the number of variables in $\boldsymbol{\alpha}$, and the second term accounts for the number of variables in $\cccov$ ($N_k^2$ represents the number of degrees of freedom of the $k$-th complex Hermitian matrix). We will later present results that will enable us to significantly reduce this number.

\section{Problems with linear rate objectives}
\label{OptGMAC:Case02}
In this section we will present results for problems with linear rate objectives. In particular, we will provide a complete description of the optimum decoding orders, and conclude that the optimum collection of covariance matrices, $\cccov$, contains one composite covariance matrix, $\ccov$. These results will be used in the following section to draw insight into problems with nonlinear objectives.

When the objective is linear, the optimization in~\eqref{eq:MACP} can be cast as
\begin{equation}\label{eq:MACPLinear}
    \underset{\boldsymbol{\alpha}\in\mathcal{S}_\alpha,\;\ccov[(m)]\in\feascov,\;m=1,\dots,K+1}{\max}\quad\sum_{k=1}^Kw_k\rho_k\left(\boldsymbol{\alpha},\cccov\right),
\end{equation}
where $\{w_k\}_{k=1}^K$ are constant, though not necessarily positive, weights. Considering non-positive weights will be useful when we consider nonlinear rate objectives in Section~\ref{NonLinear}.

With $\{w_k\}_{k=1}^K$ given, it is optimum to decode the users following the increasing order of the weights~\cite[Section~III.B]{Tse_Hanly}. Other decoding orders are optimum if they satisfy the conditions of the following proposition.

\begin{proposition}\label{theo:GMACLinearNS}\label{theo:GMACLinear}
For problems with linear rate objectives and arbitrary power constraints, a particular decoding order is optimum if and only if either
\begin{enumerate}
\item it follows the increasing order of the weights, that is, if $w_1\leq\cdots\leq w_K$, user $k$ is decoded before user $k+1$, $k=1,\dots,K-1$; or
\item the rate vector achieved with this ordering is also achieved with an ordering that follows the increasing order of the weights with the same covariance matrices.
\end{enumerate}
\end{proposition}
\begin{proof}
The proof is based on sequential comparisons of the rates achieved by pairs of orderings that are identical for all but two contiguous positions. See details in Appendix~\ref{App:Proof:GMACLinearNS}.
\end{proof}

Proposition~\ref{theo:GMACLinearNS} provides a necessary and sufficient condition that describes all the optimum decoding orderings. The first item of this proposition was also identified in~\cite{Tse_Hanly}.

\begin{remark}\label{theo:GMACLinearRemark}
Proposition~\ref{theo:GMACLinearNS} implies that
\begin{enumerate}
\item if maximizing a linear rate objective requires time-sharing of multiple rate vectors, each of these vectors can be achieved with the same decoding order; and
\item all optimum rate vectors can be achieved with the decoding orders that follow the increasing order of the weights. In other words, the second item of Proposition~\ref{theo:GMACLinearNS} implies that, if the rate vector $\{r_{ki}(\ccov)\}_{k=1}^K$ is optimum and the ordering $i$ does not follow the increasing order of the weights, then $r_{ki}(\ccov)=r_{k\opt{i}}(\ccov)$, $k=1,\dots,K$, for some ordering $\opt{i}$ that follows the increasing order of the weights.\hfill$\Box$
\end{enumerate}
\end{remark}

The first point of Remark~\ref{theo:GMACLinearRemark} implies that the number of variables can be reduced from $(K+1)!+(K+1)\sum_{k=1}^KN_k^2$ to $(K+1)(1+\sum_{k=1}^KN_k^2)$ without loss of optimality. The second item of Remark~\ref{theo:GMACLinearRemark} will be used in Section~\ref{OptGMAC:Case03} to provide a necessary condition that must be satisfied by the solutions of problems with non-linear rate objectives. It is also worth noting that if multiple weights are equal, the decoding of the respective users can be interchanged without loss of optimality.

Finding the optimum solution of~\eqref{eq:MACPLinear} is significantly facilitated by the following result.

\begin{proposition}\label{theo:OneCovMatLinObjA}
For problems with linear rate objectives and arbitrary power constraints, the objective is maximized with one composite covariance matrix, $\ccov$.
\end{proposition}
\begin{proof}
The first item of Remark~\ref{theo:GMACLinearRemark} is used to fix the decoding order in writing the expression of the optimum rate vector. The proof then uses a contradiction argument. First, it assumes that an optimum pair $(\opt{\boldsymbol{\alpha}},\cccovo)$ is achieved by time-sharing two or more composite covariance matrices. Then, it shows that one of these matrices suffices to maximize the linear rate objective. See Appendix~\ref{App:Proof:OneCovMatLinObjA}.
\end{proof}

Combining Proposition~\ref{theo:GMACLinear} and Proposition~\ref{theo:OneCovMatLinObjA}, it can be seen that, for problems with linear objectives and arbitrary power constraints, time-sharing is not necessary to solve~\eqref{eq:MACPLinear}; one decoding order and one composite covariance matrix suffice to solve~\eqref{eq:MACPLinear}. As a consequence, the number of variables in this class of problems is reduced to $\sum_{k=1}^KN_k^2$. In this case, the problem in \eqref{eq:MACPLinear} can be simplified to
\begin{equation}\label{eq:MACPLinearOneOrd}
    \underset{\ccov\in\feascov}{\max}\quad\sum_{k=1}^Kw_kr_{ki_w}\left(\ccov\right),
\end{equation}
where the $i_w$-th ordering follows the increasing order of the weights. Let $w_1\leq\cdots\leq w_K$, possibly after relabeling the users, then the objective in~\eqref{eq:MACPLinearOneOrd} can be expressed as
\begin{equation}\label{eq:prop:02}
\begin{split}
\sum_{k=1}^Kw_kr_{ki_w}\left(\ccov[]\right)=&
\sum_{k=1}^Kw_k\log\frac{\left|\boldsymbol{I}+\sum_{j\geq k}\boldsymbol{H}_j\ucov[][j]\boldsymbol{H}_j^\dagger\right|}{\left|\boldsymbol{I}+\sum_{j>k}\boldsymbol{H}_j\ucov[][j]\boldsymbol{H}_j^\dagger\right|}\\
=&\sum_{k=1}^K\!\left(w_k-w_{k-1}\right)\log\Big|\boldsymbol{I}+\sum_{j\geq
k}\boldsymbol{H}_j\ucov[][j]\boldsymbol{H}_j^\dagger\Big|.
\end{split}
\end{equation}
If the weights are nonnegative, i.e., $0\leq w_1\leq\cdots\leq w_K$, and $\feascov$ is convex, the problem in~\eqref{eq:MACPLinearOneOrd} is convex and can be solved efficiently. When some of the weights are negative or $\feascov$ is nonconvex, the problem in~\eqref{eq:MACPLinearOneOrd} is nonconvex and no solution is readly available for it. In either case,~\eqref{eq:MACPLinearOneOrd} and~\eqref{eq:prop:02} can be used to model a certain class of problems with nonlinear rate objectives as we show in Section~\ref{NonLinear}.

\section{Problems with nonlinear rate objectives}
\label{NonLinear}
We will now use the results of Section~\ref{OptGMAC:Case02} and Appendix~\ref{preliminaries} to gain insight into problems with nonlinear differentiable rate objectives, like the one in~\eqref{eq:MACP}. In particular, Propositions~\ref{theo:TransLinNec} and~\ref{theo:TransLinSuf} in Appendix~\ref{preliminaries} use variational inequalities to identify classes of problems with nonlinear objectives that share solutions with particular problems with linear objectives. These propositions will allow us to use the results of Section~\ref{OptGMAC:Case02} to describe the optimum covariance matrices, decoding orders and time-sharing parameters that reach the optimum rate vector for a broad class of problems with potentially nonlinear rate objectives.

\subsection{Characterization of the optimum transmission parameters}
\label{OptGMAC:Case03}
We begin by presenting a condition that any optimum collection of composite covariance matrices and time-sharing matrix, including those with more than $K+1$ non-zero elements, must satisfy.

\begin{theorem}\label{theo:LinNecc}\label{theo:SufNecCondGen}
Let the objective of the problem in~\eqref{eq:MACP} be continuously differentiable, and let $\opt{\boldsymbol{\alpha}}$ and $\cccovo$ be optimum. Let $w_0=0$ and let $\{w_k\}_{k=1}^K$ be given by
\begin{equation}\label{eq:theo:LinNecc:00}
    w_k=-\frac{\partial f(\boldsymbol{x})}{\partial x_k}\bigg|_{\boldsymbol{x}=\boldsymbol{\rho}\left(\opt{\boldsymbol{\alpha}},\cccovo\right)},\  k=1,\dots,K.
\end{equation}
Let the users be labelled so that $w_1\leq\cdots\leq w_K$. Then, for each strictly positive element of $\opt{\boldsymbol{\alpha}}$, say $\opt{\alpha}_{mi}$,
\begin{enumerate}
\item decoding the users following the order of the permutation $\pi_i(\cdot)$ is optimum for the linear objective defined by the weights $\{w_k\}_{k=1}^K$ in~\eqref{eq:theo:LinNecc:00}, cf. Proposition~\ref{theo:GMACLinearNS}, and
\item the composite covariance matrix $\ccovo[(m)]$ solves
\end{enumerate}
\begin{equation}\label{eq:theo:LinNecc:01}
    \underset{\ccov[]\in\feascov}{\max}\ \displaystyle\sum_{k=1}^K(w_k-w_{k-1})\log\Bigl|\boldsymbol{I}+\sum_{j\geq k}\boldsymbol{H}_j\ucov[][j]\boldsymbol{H}_j^\dagger\Bigr|.
\end{equation}

If, in addition to continuous differentiability, the objective in~\eqref{eq:MACP} is convex in the users' rates, then every pair $(\opt{\boldsymbol{\alpha}},\cccovo)$ that satisfies the previous two conditions is optimum.
\end{theorem}

\begin{proof}
The proof of necessity for continuously differentiable objectives hinges on the convexity of the GMAC capacity region and Proposition~\ref{theo:TransLinNec} in Appendix~\ref{preliminaries}, whereas the proof of sufficiency for convex objectives follows from applying Proposition~\ref{theo:TransLinSuf} in Appendix~\ref{preliminaries}. See Appendix~\ref{App:Proof:theo:LinNecc} for details.
\end{proof}

Theorem~\ref{theo:SufNecCondGen} provides a necessary optimality condition for problems with continuously differentiable objectives, which is also sufficient for objectives that are additionally convex in the users' rates, but not necessarily convex in $\boldsymbol{\alpha}$ and $\cccov$.

\begin{remark}

$\ $

\begin{enumerate}
\item The weights in~\eqref{eq:theo:LinNecc:00} can be negative, and, for given weights, the conditions of this theorem are decoupled from each other. In particular, the first condition does not involve the covariance matrices, and the second condition does not involve the time-sharing parameters. This observation will be used by the algorithm proposed in Section~\ref{subsec:alg} to find the optimum covariance matrices, decoding orders and time-sharing parameters.

\item Theorem~\ref{theo:SufNecCondGen} does not impose any constraints on $\feascov$. For instance, $\feascov$ can be nonconvex.

\item If the solution of~\eqref{eq:MACP} is unique, the optimum time-sharing matrix $\opt{\boldsymbol{\alpha}}$ has no more than $K+1$ non-zero elements. If, however, the solution of~\eqref{eq:MACP} is not unique, some of the solutions described by Theorem~\ref{theo:SufNecCondGen} may have time-sharing matrices with more than $K+1$ non-zero elements.\hfill$\Box$
\end{enumerate}
\end{remark}

Although Theorem~\ref{theo:SufNecCondGen} provides an explicit characterization of the optimum transmission parameters, it cannot be readily used to obtain these parameters. This is because the gradient of the objective at the optimum rate vector, and hence the weights in~\eqref{eq:theo:LinNecc:00}, are not known \emph{a priori}.

When the gradient at the optimum rate vector is given, finding the optimum transmission parameters of~\eqref{eq:MACP} might be still complex. In particular, when~\eqref{eq:theo:LinNecc:01} has multiple solutions, it is not known \emph{a priori} which solutions should be included in $\cccov$. As such, it is necessary for all possible collections of solutions of~\eqref{eq:theo:LinNecc:01} to be checked for optimality. In general, one optimum solution of~\eqref{eq:theo:LinNecc:01} does not suffice to achieve the optimum transmission parameters of~\eqref{eq:MACP}, as we will show in the following example.

\subsection{An illustrative example}
\label{OptGMAC:FirstExample}
In this section, we present a case in which~\eqref{eq:theo:LinNecc:01} has multiple solutions, none of which alone suffices to achieve the optimum rate vector.

Let us consider a GMAC with two receive antennas and two users with two transmit antennas each. The channel matrices of the users are assumed to be full rank and identical, that is, $\boldsymbol{H}_1=\boldsymbol{H}_2=\boldsymbol{H}$ and $|\boldsymbol{H}|\neq0$. Each transmitter is allowed to transmit using only one of the antennas, and with a maximum power $P$. To maximize proportional fairness~\cite{Kelly}, the optimum covariance matrices and time-sharing matrix must solve
\begin{multline}\label{eq:Example_init:Problem}
    \underset{\boldsymbol{\alpha}\in\mathcal{S}_\alpha,\;\ccov[(m)]\in\feascov,\;m=1,2,3}{\min}\     f\big(\boldsymbol{\rho}(\boldsymbol{\alpha},\cccov)\big)=\\-\log\big(\rho_1(\boldsymbol{\alpha},\cccov)\big)-\log\big(\rho_2(\boldsymbol{\alpha},\cccov)\big),
\end{multline}
\begin{multline}\label{eq:Example_init:P}
\feascov=\{\ccov[]\in\mathcal{D}\;|\;\ccov[]\succeq0,\;\text{tr}(\ucov[][1])\leq P,\;\text{tr}(\ucov[][2])\leq P,\\
[\ucov[][1]]_{11}[\ucov[][1]]_{22}=0,\;[\ucov[][2]]_{11}[\ucov[][2]]_{22}=0\},
\end{multline}
where $\rho_k(\boldsymbol{\alpha},\cccov)$, $k=1,2$, are defined in~\eqref{eq:rate_k}. The nonconvex constraints $[\ucov[][1]]_{11}[\ucov[][1]]_{22}=0$ and $[\ucov[][2]]_{11}[\ucov[][2]]_{22}=0$, and the positive definiteness of $\ucov[][1]$ and $\ucov[][2]$ imply that $\ccov[]$ is diagonal. Since the objective, $f$, in~\eqref{eq:Example_init:Problem} is convex in the users' rates, Theorem~\ref{theo:SufNecCondGen} can be used to characterize the optimum transmission parameters.

To apply this theorem, first, we will find the gradient at the optimum rate vector; second, we will find all optimum composite covariance matrices that solve~\eqref{eq:theo:LinNecc:01}; and third, we will analyze which of these solutions have to be included in the collection of composite covariance matrices, $\cccov$.

To find the gradient at the optimum rate vector, let $\opt{\boldsymbol{x}}=(\opt{x_1},\opt{x_2})^\intercal=\boldsymbol{\rho}(\opt{\boldsymbol{\alpha}},\cccovo)$ be optimum for~\eqref{eq:Example_init:Problem}. From symmetry, the rate vector $(\opt{x_2},\opt{x_1})^\intercal$ is feasible. Now, since the capacity region is convex, Proposition~\ref{theo:TransLinNec} in Appendix~\ref{preliminaries} asserts that ${\opt{\boldsymbol{x}}}^\intercal\nabla f(\opt{\boldsymbol{x}})\leq(\opt{x_2},\opt{x_1})\nabla f(\opt{\boldsymbol{x}})$, which, using the fact that $\nabla f(\opt{\boldsymbol{x}})=-\big(\frac{1}{\opt{x_1}},\frac{1}{\opt{x_2}}\big)^\intercal$, yields $(\opt{x_1}-\opt{x_2})^2\leq0$. This implies that $\opt{x_1}=\opt{x_2}$ and that the gradient components at the optimum rate vector are equal.

Using the fact that the gradient components are equal at the optimum rate vector and invoking the second condition of Theorem~\ref{theo:SufNecCondGen}, it can be seen that the optimum composite covariance matrices solve
\begin{multline}\label{eq:Example_init:Problem_theo4}
    \underset{\ccov[]\in\feascov}{\max}\ \big([(\boldsymbol{H}^\dagger\boldsymbol{H})^{-1}]_{11}+[\ucov[][1]]_{11}+[\ucov[][2]]_{11}\big)\\
    \big([(\boldsymbol{H}^\dagger\boldsymbol{H})^{-1}]_{22}+[\ucov[][1]]_{22}+[\ucov[][2]]_{22}\big).
\end{multline}
The previous problem is not convex because of the nonconvexity of $\feascov$, cf.~\eqref{eq:Example_init:P}. Taking into account the constraints in~\eqref{eq:Example_init:Problem_theo4}, it is straightforward to show that only four points can be optimum, which are those for which the quadruple $\big([\ucov[][1]]_{11},[\ucov[][1]]_{22},[\ucov[][2]]_{11},[\ucov[][2]]_{22}\big)$ equals: 1) $(P,0,0,P)$; 2) $(0,P,P,0)$; 3) $(P,0,P,0)$; and 4) $(0,P,0,P)$. In our example we assume
\begin{subequations}\label{eq:Example_init:gammas}
\begin{equation}
{\color{black}P>\big|[(\boldsymbol{H}^\dagger\boldsymbol{H})^{-1}]_{11}\!-\![(\boldsymbol{H}^\dagger\boldsymbol{H})^{-1}]_{22}\big|,}
\end{equation}
\begin{equation}
{\color{black}P>\frac{\big([(\boldsymbol{H}^\dagger\boldsymbol{H})^{-1}]_{11}\!-\![(\boldsymbol{H}^\dagger\boldsymbol{H})^{-1}]_{22}\big)\big|(\boldsymbol{H}^\dagger\boldsymbol{H})^{-1}\big|}{\big|(\boldsymbol{H}^\dagger\boldsymbol{H})^{-1}\big|-[(\boldsymbol{H}^\dagger\boldsymbol{H})^{-1}]_{11}^2}.}
\end{equation}
\end{subequations}
It can be readily verified that, under the first assumption, only the first two quadruples solve~\eqref{eq:Example_init:Problem_theo4}. Moreover, we will show that, even time-sharing the $K!=2$ decoding orders, neither quadruple alone satisfies Theorem~\ref{theo:SufNecCondGen}. In other words, to achieve the optimum rate vector, two quadruples must be time-shared with their corresponding decoding orders. Towards that end, we will show that, subject to~\eqref{eq:Example_init:gammas}, the two users cannot have equal rates. We will focus on the first quadruple $(P,0,0,P)$; an analogous argument can be applied to the second quadruple $(0,P,P,0)$.

Let the diagonal of $\ccovo[]$ be given by the first quadruple $(P,0,0,P)$. The rates achievable with this $\ccovo[]$ are $x_k=a_1r_{k1}\big(\ccovo[]\big)+a_2r_{k2}\big(\ccovo[]\big)$, $k=1,2$, where $a_1$ and $a_2$ are the time-sharing weights. Let the first decoding order be $\pi_1(1)=1$ and $\pi_1(2)=2$, that is, user 1 is decoded first, and its signal is interfered by that of user 2. Then, from~\eqref{eq:r_ki} and the inequality in~\eqref{eq:theo:OneCovMatFixed_Prop:ineq} in Appendix~\ref{App:Proof:GMACLinear}, it follows that $r_{12}\big(\ccovo[]\big)\geq r_{11}\big(\ccovo[]\big)$ and $r_{21}\big(\ccovo[]\big)\geq r_{22}\big(\ccovo[]\big)$. Using the first quadruple $(P,0,0,P)$ in~\eqref{eq:r_ki} for the rates $r_{11}$ and $r_{21}$, it can be shown that the second assumption of~\eqref{eq:Example_init:gammas} implies that $r_{11}\big(\ccovo[]\big)>r_{21}\big(\ccovo[]\big)$. Thus, for any time-sharing weights $a_1$ and $a_2$, the rate of user $1$ is strictly greater than the rate of user $2$, i.e., $x_1>x_2$. This implies that the optimum rate vector can be achieved only by time-sharing the first two quadruples, but not with one of them only.

\subsection{Two special cases}
\label{OptGMAC:Case04}
The illustrative example in the previous section showed that, in general, the optimum rate vector of~\eqref{eq:MACP} requires time-sharing of multiple composite covariance matrices and decoding orders, which necessitates finding all the solutions of~\eqref{eq:theo:LinNecc:01}, thereby complicating the solution of~\eqref{eq:MACP}. In order to avoid this difficulty, this section identifies two classes of problems for which~\eqref{eq:theo:LinNecc:01} has a unique solution. In both cases, the power constraints, $g_\ell$ in~\eqref{eq:P_set}, $\ell=1,\ldots,L$, are convex in the input covariance matrices, which ensures that the feasible set $\feascov$ in~\eqref{eq:theo:LinNecc:01} is convex. For these classes, the search for a collection of composite covariance matrices, $\cccov$, reduces to the search for only one composite covariance matrix, $\ccov[]$.

\paragraph*{Case I) Nonincreasing objective with respect to the users' rates}
If the objective function, $f$, is nonincreasing in each component, its gradient components are always nonpositive, which implies that the weights computed in~\eqref{eq:theo:LinNecc:00} are nonnegative. This fact ensures that the objective in~\eqref{eq:theo:LinNecc:01} is concave and, together with the convexity of the power constraints, that the optimization problem in~\eqref{eq:theo:LinNecc:01} is convex, so that it can be solved with convex optimization tools. Since the logarithm of the determinant is a strictly concave function, cf.~\cite[p.~74]{Boyd}, the optimum solution is unique.

\paragraph*{Case II)}
Suppose that the power constraints are convex and that $\ucov[][1]\oplus\cdots\oplus\ucov[][K]\in\feascov$ implies that $\ucov[][1]\oplus\cdots\oplus\ucov[][k-1]\oplus\boldsymbol{0}\oplus\ucov[][k+1]\oplus\cdots\oplus\ucov[][K]\in\feascov$, $k=1,\dots,K$. In other words, setting any one of the individual covariance matrices to zero is also feasible. In this case, the weights computed in~\eqref{eq:theo:LinNecc:00} may be negative for some differentiable objectives, which implies that in those cases the problem in~\eqref{eq:theo:LinNecc:01} is not convex. However, since the optimum covariance matrix of users with negative weights is zero, those users can be eliminated from the formulation in~\eqref{eq:theo:LinNecc:01}, resulting in an optimization problem with only nonnegative weights, subsequently, with a unique solution.

\subsection{Problems with one optimum composite covariance matrix}
\label{OptGMAC:Case04b}
In this section, we will show that, when~\eqref{eq:theo:LinNecc:01} has a unique solution, as in the cases described in the previous section, Theorem~\ref{theo:LinNecc} can be used to reduce the complexity of solving and analyzing~\eqref{eq:MACP}. This fact is stated in the following result.

\begin{theorem}\label{theo:OneCovMatGen}
If the objective $f$ in~\eqref{eq:MACP} is continuously differentiable in the users' rates, and if the solution of~\eqref{eq:theo:LinNecc:01} is unique, the optimization problem in~\eqref{eq:MACP} can be solved with one composite covariance matrix, $\ccovo[]$.
\end{theorem}
\begin{proof}
Let $(\opt{\boldsymbol{\alpha}},\cccovo)$ solve~\eqref{eq:MACP}, where $\cccovo=\{\ccovo[(m)]\}_{m=1}^{K+1}$. Using this solution, the weights in~\eqref{eq:theo:LinNecc:00} can be readily obtained. Let $\ccovno[]$ be the unique solution of~\eqref{eq:theo:LinNecc:01} for these weights. Then, if $\opt{\alpha}_{mi}>0$ for some ordering $i$, the second condition of Theorem~\ref{theo:LinNecc} implies that $\ccovo[(m)]=\ccovno[]$. Moreover, if $\opt{\alpha}_{mi}=0$ for all orderings $i=1,\dots,K!$, the definition in~\eqref{eq:rate_k} implies that $\ccovo[(m)]$ does not affect the rate vector. As a consequence, the pair $\big(\opt{\boldsymbol{\alpha}},\{\ccovno[]\}_{m=1}^{K+1}\big)$ generates the same rate vector, and hence solves~\eqref{eq:MACP}.
\end{proof}

Theorem~\ref{theo:OneCovMatGen} implies that, when~\eqref{eq:theo:LinNecc:01} has a unique solution, the optimum rate vector of~\eqref{eq:MACP} does not need to be expressed as a function of a time-sharing matrix and a collection of composite covariance matrices. In particular, the time sharing matrix $\boldsymbol{\alpha}$ can be replaced with a time-sharing vector $\boldsymbol{\beta}\in\mathbb{R}^{K!}$ that lies in the unit $K!$-dimensional simplex, $\mathcal{S}_\beta=\{\boldsymbol{\beta}|\boldsymbol{\beta}\geq0,\sum_{i=1}^{K!}\beta_i=1\}$. Using this notation, achievable rate vectors in~\eqref{eq:rate_k} can be expressed as $\hat{\boldsymbol{\rho}}\bigl(\boldsymbol{\beta},\ccov[]\bigr)$, where
\begin{equation}\label{eq:RVC}
\hat\rho_k\left(\boldsymbol{\beta},\ccov[]\right)=\sum_{i=1}^{K!}\beta_ir_{ki}\left(\ccov[]\right),\quad k=1,\ldots,K.
\end{equation}
Notice that the number of variables in~\eqref{eq:RVC} is reduced to $K!+\sum_kN_k^2$. For the two cases of Section~\ref{OptGMAC:Case04}, the optimum rate vector of~\eqref{eq:MACP} can be expressed using this notation. In Sections~\ref{Application} and~\ref{examples} we will consider the cases of Section~\ref{OptGMAC:Case04}, and use the characterization of~\eqref{eq:RVC}.

\section{Algorithm description and convergence analysis}
\label{Application}
In this section, we develop and analyze the convergence of an algorithm that solves a class of the not necessarily convex GMAC optimization problems in~\eqref{eq:MACP}. This class includes the two cases in Section~\ref{OptGMAC:Case04}, wherein the objective $f$ is continuously differentiable and convex in the users' rates, but not necessarily convex in the input covariance matrices and the time-sharing parameters. We will prove the convergence of the algorithm in Section~\ref{subsec:convergence} when $f$ is twice differentiable and nonincreasing in each component.

For the cases in Section~\ref{OptGMAC:Case04}, the solution of~\eqref{eq:theo:LinNecc:01} is unique and Theorem~\ref{theo:OneCovMatGen} enables us to use the notation in~\eqref{eq:RVC} and express~\eqref{eq:MACP} as
\begin{equation}\label{eq:MACP2}
\underset{\boldsymbol{\beta}\in\mathcal{S}_\beta,\ccov[]\in\feascov}{\min}\quad f\big(\hat{\boldsymbol{\rho}}(\boldsymbol{\beta},\ccov[])\big).
\end{equation}
To solve this problem we will develop an algorithm based on Theorem~\ref{theo:SufNecCondGen}.

\subsection{Proposed algorithm}\label{subsec:alg}
The proposed algorithm solves~\eqref{eq:MACP2} iteratively. At each iteration $t$, the algorithm uses $\ccovo(t-1)$ and $\opt{\boldsymbol{\beta}}(t-1)$, obtained at the previous iteration, to obtain $\ccovo(t)$ and $\opt{\boldsymbol{\beta}}(t)$. In particular, the algorithm uses~\eqref{eq:theo:LinNecc:00} to compute the weights $\{w_k(t)\}_{k=1}^K$ at the rate vector $\boldsymbol{x}=\hat{\boldsymbol{\rho}}\big(\opt{\boldsymbol{\beta}}(t-1),\ccovo(t-1)\big)$. Subsequently, Theorem~\ref{theo:SufNecCondGen} is used to find the time-sharing vector $\opt{\boldsymbol{\beta}}(t)$ and the composite covariance matrix $\ccovo[](t)$. We will show, in Section~\ref{subsec:convergence}, that $f\big(\hat{\boldsymbol{\rho}}\big(\opt{\boldsymbol{\beta}}(t),\ccovo(t)\big)\big)$ converges to the optimum solution of~\eqref{eq:MACP2}.

To complete the description of the algorithm, it remains to show how $\opt{\boldsymbol{\beta}}(t)$ and $\ccovo[](t)$ are obtained. We begin by obtaining $\ccovo[](t)$. To do that, let $\ccovno[](t)$ be the solution of~\eqref{eq:theo:LinNecc:01} for the weights $\{w_k(t)\}_{k=1}^K$, i.e.,
\begin{equation}\label{eq:Alg:LP_Qt}
    \underset{\ccov[]\in\feascov}{\max}\ \displaystyle\sum_{k=1}^K\big(w_k(t)-w_{k-1}(t)\big)\log\Bigl|\boldsymbol{I}+\sum_{j\geq k}\boldsymbol{H}_j\ucov[][j]\boldsymbol{H}_j^\dagger\Bigr|,
\end{equation}
where the users are relabeled such that $w_1(t)\leq\cdots\leq w_K(t)$ and $w_0(t)=0$ for all $t$. We will choose $\ccovo[](t)$ to be a convex combination of $\ccovo[](t-1)$ and $\ccovno[](t)$. In particular,
\begin{equation}\label{eq:Alg:cov}
    \ccovo(t)=\varepsilon(t)\ccovno(t)+\big(1-\varepsilon(t)\big)\ccovo(t-1),
\end{equation}
where $\varepsilon(t)$ is selected as the Armijo stepsize rule~\cite{Meisam_Convergence}, i.e., for some $\sigma,\chi\in(0,1)$, $\varepsilon(t)$ is chosen to be the largest element of $\{\chi^n\}_{n=0,1,\dots}$ satisfying
\begin{equation}\label{eq:Alg:epsilon}
h(\varepsilon(t),t)\leq h(0,t)+\sigma\varepsilon(t)\frac{\partial h(\delta,t)}{\partial\delta}\bigg|_{\delta=0},
\end{equation}
where
\begin{equation}\label{eq:Alg:h}
h(\delta,t)=f\big(\hat{\boldsymbol{\rho}}\big(\opt{\boldsymbol{\beta}}(t-1),\delta\ccovno[](t)+(1-\delta)\ccovo(t-1)\big)\big). \end{equation}
The following result shows that the objective function $f$ does not increase using this stepsize.

\begin{proposition}\label{theo:Alg:Stepsize}
Let the objective function $f$ be nonincreasing in each component. Then, for every iteration $t$ we have that
\begin{equation}\label{eq:theo:Alg:f_decrease}
    f\big(\hat{\boldsymbol{\rho}}\big(\opt{\boldsymbol{\beta}}(t-1),\ccovo(t)\big)\big)\leq f\big(\hat{\boldsymbol{\rho}}\big(\opt{\boldsymbol{\beta}}(t-1),\ccovo(t-1)\big)\big).
\end{equation}
\end{proposition}
\begin{proof}
See Appendix~\ref{App:Proof:Alg:Stepsize}.
\end{proof}

After computing $\ccovo[](t)$, the time-sharing vector $\opt{\boldsymbol{\beta}}(t)$ is chosen to solve the following convex optimization problem:
\begin{equation}\label{eq:Algorithm:Prob02}
    \underset{\boldsymbol{\beta}\in\mathcal{S}_\beta}{\min}\quad f\big(\hat{\boldsymbol{\rho}}\big(\boldsymbol{\beta},\ccovo[](t)\big)\big).
\end{equation}
Hence, it can be shown that
\begin{equation}\label{eq:theo:Alg:f_decrease2}
    f\big(\hat{\boldsymbol{\rho}}\big(\opt{\boldsymbol{\beta}}(t),\ccovo(t)\big)\big)\leq f\big(\hat{\boldsymbol{\rho}}\big(\opt{\boldsymbol{\beta}}(t-1),\ccovo(t)\big)\big).
\end{equation}

Proposition~\ref{theo:Alg:Stepsize} and~\eqref{eq:theo:Alg:f_decrease2} imply that the value of the objective at the end of each iteration is equal to or less than its value at the end of the previous iteration. The steps of this algorithm are summarized in the pseudocode of Algorithm~\ref{Alg:Algorithm}.
\begin{algorithm}[b]
\caption{Computation of the optimum transmission parameters of~\eqref{eq:MACP2}.}
\label{Alg:Algorithm}
\begin{algorithmic}[1]
\STATE Initialize $\opt{\boldsymbol{\beta}}(0)$ and $\ccovo[](0)$.\\
\FOR{$t=1,2,\dots$}
    \STATE Use~\eqref{eq:theo:LinNecc:00} to compute $\{w_k(t)\}_{k=1}^K$ at the rate vector $\boldsymbol{x}=\hat{\boldsymbol{\rho}}\big(\opt{\boldsymbol{\beta}}(t-1),\ccovo(t-1)\big)$.\\
    \STATE Search for the optimum composite covariance matrix, $\ccovno(t)$, that solves~\eqref{eq:theo:LinNecc:01} for the weights $\{w_k(t)\}_{k=1}^K$.
    \STATE Use~\eqref{eq:Alg:cov} and~\eqref{eq:Alg:epsilon} to find the optimum composite covariance matrix at iteration $t$, $\ccovo(t)$.
    \STATE Solve~\eqref{eq:Algorithm:Prob02} to obtain $\opt{\boldsymbol{\beta}}(t)$.
\ENDFOR
\end{algorithmic}
\end{algorithm}

Before studying the convergence of this algorithm, we will show that it can be readily used to generate lower bounds on the objective. In particular, we have the following result.

\begin{proposition}\label{theo:Alg:Error}
Let $\opt{\boldsymbol{x}}$ be the optimum rate vector. Then, the following inequalities hold:
\begin{multline}\label{eq:Algorithm:eq02}
    f\big(\hat{\boldsymbol{\rho}}\big(\opt{\boldsymbol{\beta}}\!(t),\ccovo\!(t)\big)\big)\!+\!
    \sum_{k=1}^K\!w_k(t\!+\!1)\!\Big(\hat\rho_k\!\big(\opt{\boldsymbol{\beta}}\!(t),\ccovo\!(t)\big)-\\
    \hat\rho_k\!\big(\opt{\boldsymbol{\beta}}\!(t),\ccovno\!(t\!+\!1)\big)\!\Big)\leq
    f(\opt{\boldsymbol{x}})\leq f\big(\hat{\boldsymbol{\rho}}\big(\opt{\boldsymbol{\beta}}\!(t),\ccovo\!(t)\big)\big).
\end{multline}
\end{proposition}
\begin{proof}
The second inequality follows from the optimality of $\opt{\boldsymbol{x}}$ and the feasibility of $\opt{\boldsymbol{\beta}}(t)$ and $\ccovo(t)$. See Appendix~\ref{App:Proof:Alg:Error} for the proof of the first inequality.
\end{proof}

This proposition implies that the left hand side of~\eqref{eq:Algorithm:eq02} is a lower bound on the objective value at the optimum, and can hence be used as a stopping criterion for Algorithm~\ref{Alg:Algorithm}.

\subsection{Convergence analysis}\label{subsec:convergence}
In this section, we will show that Algorithm~\ref{Alg:Algorithm} is guaranteed to converge to the global optimal of~\eqref{eq:MACP2}, provided that the conditions of the following theorem are satisfied.

\begin{theorem}\label{theo:Convergence}
Let Algorithm~\ref{Alg:Algorithm} be used to solve the optimization problem in~\eqref{eq:MACP2} and suppose that the following conditions are satisfied:
\begin{enumerate}
\item The power constraint set $\feascov$ is convex in $\ccov$;
\item The objective $f$ is second order differentiable, nonincreasing in each component, and convex in the users' rates, $\hat{\boldsymbol{\rho}}$, but not necessarily convex in the time-sharing vector, $\boldsymbol{\beta}$, and the composite covariance matrix, $\ccov$.
\end{enumerate}
It follows that Algorithm~\ref{Alg:Algorithm} converges to the optimum time-sharing vector and the optimum composite covariance matrix, that is, if $\opt{\boldsymbol{x}}$ is the optimum rate vector, then
\begin{equation}
\lim_{t\rightarrow\infty}f\big(\hat{\boldsymbol{\rho}}\big(\opt{\boldsymbol{\beta}}(t),\ccovo(t)\big)\big)-f(\opt{\boldsymbol{x}})=0.
\end{equation}
\end{theorem}
\begin{proof}
The proof uses the bound in Proposition~\ref{theo:Alg:Error} as $t$ goes to $\infty$. See details in Appendix~\ref{App:Proof:Convergence}.
\end{proof}

The merit of Theorem~\ref{theo:Convergence} is that identifies a class of, potentially intricate, problems for which Algorithm~\ref{Alg:Algorithm} converges to the optimum solution. In Section~\ref{examples} we will use this algorithm to solve two instances of this class of problems.

\section{Numerical examples}\label{examples}
This section provides two instances of the problem in~\eqref{eq:MACP} to illustrate the applicability of the information theoretic results of Section~\ref{NonLinear} in practical communication scenarios. In both instances, Algorithm~\ref{Alg:Algorithm} is used to solve a two-user GMAC optimization problem with an objective that is convex in the users rates but nonconvex in the transmission parameters, and each user is assumed to have two transmit antennas and a power budget of $P=10$~dB. For both instances the destination has two receive antennas and the channel matrices are given by
\begin{equation}
    \boldsymbol{H}_1=\left(
    \begin{array}{cc}
        0.32  & -0.06\\
        -0.72 & -0.88
    \end{array}
    \right)+\jmath\left(
    \begin{array}{cc}
        -0.15 & -1.38\\
        -1.34 & -0.01
    \end{array}
    \right),
\end{equation}
\begin{equation}
    \boldsymbol{H}_2=\left(
    \begin{array}{cc}
        -0.21 & 0.29\\
        -0.08  & 0.91
    \end{array}
    \right)+\jmath\left(
    \begin{array}{cc}
        -0.65  & 0.15\\
        0.13 & 1.39
    \end{array}
    \right).
\end{equation}
For these channel matrices, the GMAC capacity region is the one contained inside the thick solid and dotted lines in Figure~\ref{fig:Capacity} for the instance considered in Example~1. The same region is shown in Figure~\ref{fig:Capacity2} for the instance considered in Example~2.
\begin{figure}
\includegraphics[width=\columnwidth]{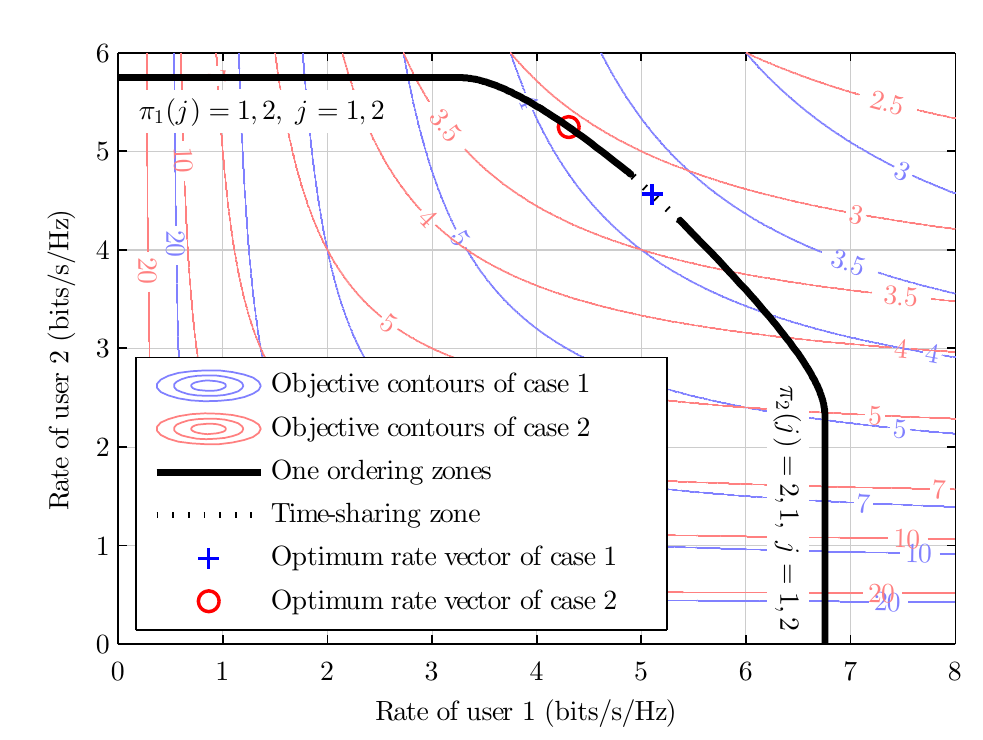}
\caption{The capacity region and the optimum rate vectors corresponding to the two cases considered in Example~1. The time-sharing zone is highlighted with a dotted line.}
\label{fig:Capacity}
\end{figure}

\subsection{Example 1}\label{subsec:app}
In this example, the objective is to minimize the total completion time~\cite{CompletionTime}, that is, the time required to transmit the data stored in the buffers of the two users. Let $b_1$ and $b_2$  represent the amount of data stored in the buffers of user~1 and~2, respectively~\cite{CompletionTime}. In this case, the optimization problem that yields the smallest total completion time can be expressed as
\begin{subequations}\label{eq:Example:Problem}
\begin{align}
    \underset{\boldsymbol{\beta},\ccov}{\min}\quad  &\frac{b_1}{\hat\rho_1\left(\boldsymbol{\beta},\ccov\right)}+\frac{b_2}{\hat\rho_2\left(\boldsymbol{\beta},\ccov\right)},\label{eq:Example:Problem:Obj}\\
    \text{subject to}\quad &\beta_1+\beta_2=1,\;\beta_i\geq0,\;i=1,2,\\
    &\boldsymbol{Q}_k\succeq0,\;\lambda\leq\text{tr}\left(\boldsymbol{Q}_k\right)\leq P,\;k=1,2,\label{eq:Example:Problem:Const03}
\end{align}
\end{subequations}
where $\lambda>0$ is a small scalar that ensures that the objective second order derivative is bounded, and hence the conditions of Theorem~\ref{theo:Convergence} are satisfied. Subsequently, Algorithm~\ref{Alg:Algorithm} converges to the optimum solution of~\eqref{eq:Example:Problem}. We note that the objective in this problem is highly nonconvex in $\boldsymbol{\beta}$ and $\ccov$. However, since the objective function in~\eqref{eq:Example:Problem:Obj} is convex in the users' rates, Theorem~\ref{theo:SufNecCondGen} can be used to characterize the optimum transmission parameters. Moreover, since the objective is nonincreasing in the users' rates, this problem represents an instance of case I in Section~\ref{OptGMAC:Case04}. From Theorem~\ref{theo:OneCovMatGen}, it follows that only one composite covariance matrix is needed, so that the formulation used in~\eqref{eq:Example:Problem} suffices to reach the optimum rate vector.

We consider two cases: 1) $(b_1,b_2)=(10,8)$; and 2) $(b_1,b_2)=(5,10)$. The contour lines of the objective for the two cases are depicted in
Figure~\ref{fig:Capacity}, and the optimum rate vectors are marked by the symbols `$+$' and `$\circ$', respectively. In case~1, the optimum rate vector lies in the time-sharing zone and the optimum covariance matrices can be obtained with either Algorithm~\ref{Alg:Algorithm} or the iterative water-filling algorithm~\cite{Wei_Yu_IWFA_MAC}. In case~2, the optimum covariance matrices are obtained with Algorithm~\ref{Alg:Algorithm}; the iterative water-filling algorithm cannot achieve these matrices in this case.

The optimum time-sharing parameters for case~1 are $\beta_1=0.57$ and $\beta_2=0.43$, and for case~2 are $\beta_1=1$ and $\beta_2=0$, i.e., for the latter case it is optimum to decode user $1$ first without time-sharing. The optimum covariance matrices for case~1 are given by
\begin{equation}
    \boldsymbol{Q}_1=\left(
    \begin{array}{cc}
        1.98 & 0.71\\
        0.71 & 8.02
    \end{array}
    \right)+\jmath\left(
    \begin{array}{cc}
        0 & -3.92\\
        3.92 & 0
    \end{array}
    \right),
\end{equation}
\begin{equation}
    \boldsymbol{Q}_2=\left(
    \begin{array}{cc}
        0.86 & 0.55\\
        0.55 & 9.14
    \end{array}
    \right)+\jmath\left(
    \begin{array}{cc}
        0 & 2.75\\
        -2.75 & 0
    \end{array}
    \right),
\end{equation}
and for case~2 are given by
\begin{equation}
    \boldsymbol{Q}_1=\left(
    \begin{array}{cc}
        1.89 & 0.64\\
        0.64 & 8.11
    \end{array}
    \right)+\jmath\left(
    \begin{array}{cc}
        0 & -3.87\\
        3.87 & 0
    \end{array}
    \right),
\end{equation}
\begin{equation}
    \boldsymbol{Q}_2=\left(
    \begin{array}{cc}
        1.39 & 0.34\\
        0.34 & 8.61
    \end{array}
    \right)+\jmath\left(
    \begin{array}{cc}
        0 & 1.83\\
        -1.83 & 0
    \end{array}
    \right).
\end{equation}

For case~1, the algorithm converged in two iterations, and for case~2 it converged in eight iterations. 
In both cases, the algorithm was stopped when the upper bound on the error in~\eqref{eq:theo:Convergence:01} dropped below~$10^{-5}$.

\subsection{Example 2}\label{subsec:app2}
In this example, the goal is to maximize the weighted proportional fairness~\cite{Kelly}, provided that the sum rate exceeds a given threshold $R$. In this case, the optimum time-sharing vector and covariance matrices are those that solve
\begin{subequations}\label{eq:Example2:Problem}
\begin{align}
    \underset{\boldsymbol{\beta},\ccov}{\max}\quad  &v_1\log\big(\hat\rho_1(\boldsymbol{\beta},\ccov)\big)+v_2\log\big(\hat\rho_2(\boldsymbol{\beta},\ccov)\big),\label{eq:Example2:Problem:Obj}\\
    \text{subject to}\quad &\beta_1+\beta_2=1,\;\beta_i\geq0,\;i=1,2,\\
    &\boldsymbol{Q}_k\succeq0,\;\lambda\leq\text{tr}\left(\boldsymbol{Q}_k\right)\leq P,\;k=1,2,\label{eq:Example2:Problem:Const02}\\
    &\log\left|\boldsymbol{I}+\boldsymbol{H}_1\boldsymbol{Q}_1\boldsymbol{H}_1^\dag+\boldsymbol{H}_2\boldsymbol{Q}_2\boldsymbol{H}_2^\dag\right|\geq R,\label{eq:Example2:Problem:Const03}
\end{align}
\end{subequations}
where $v_1>0$ and $v_2>0$ are arbitrary weights and $\lambda>0$, as in the previous example, is a small scalar that ensures that the conditions of Theorem~\ref{theo:Convergence} are satisfied. We note that the objective in this problem, as in the previous one, is highly nonconvex in $\boldsymbol{\beta}$ and $\ccov$. However, since the objective is to maximize a concave function in the users' rates, Theorem~\ref{theo:SufNecCondGen} can be used to characterize the optimum transmission parameters. Moreover, this objective is nondecreasing, which implies that, from case I in Section~\ref{OptGMAC:Case04} and Theorem~\ref{theo:OneCovMatGen} in Section~\ref{OptGMAC:Case04b}, the optimum rate vector can be reached using one composite covariance matrix. In order to use Algorithm~\ref{Alg:Algorithm} to solve this problem, the convex constraints~\eqref{eq:Example2:Problem:Const02} and~\eqref{eq:Example2:Problem:Const03} were incorporated in the formulation of the problem in~\eqref{eq:theo:LinNecc:01}, which was then solved using the CVX package~\cite{cvx} with an underlying interior point method.

\begin{figure}
\includegraphics[width=\columnwidth]{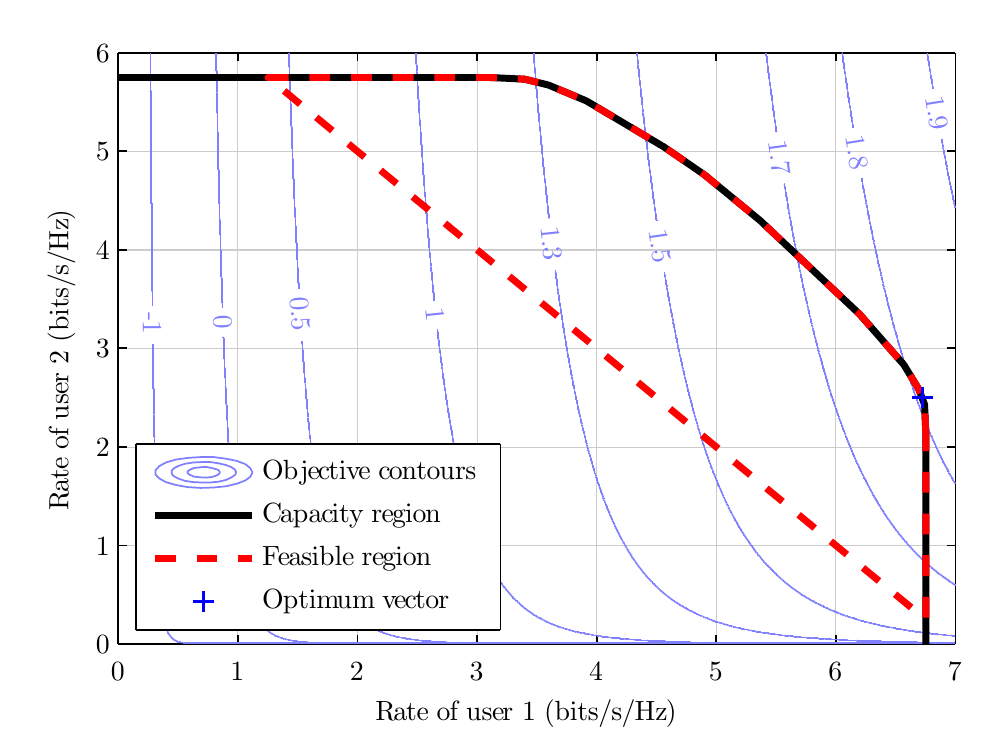}
\caption{Feasible rate region and the optimum rate vector corresponding to Example~2.}
\label{fig:Capacity2}
\end{figure}
Figure~\ref{fig:Capacity2} shows the feasible rate region of this example when $R=7$~bits/s/Hz and the optimum rate vector for $v_1=0.9$ and $v_2=0.1$. The optimum time-sharing parameters are $\beta_1=0$ and $\beta_2=1$, i.e., it is optimum to decode user $2$ first without time-sharing. The optimum covariance matrices are given by
\begin{equation}
    \boldsymbol{Q}_1=\left(
    \begin{array}{cc}
        4.46 & 0.36\\
        0.36 & 5.54
    \end{array}
    \right)+\jmath\left(
    \begin{array}{cc}
        0 & -1.13\\
        1.13 & 0
    \end{array}
    \right),
\end{equation}
\begin{equation}
    \boldsymbol{Q}_2=\left(
    \begin{array}{cc}
        0.90 & 0.56\\
        0.56 & 9.10
    \end{array}
    \right)+\jmath\left(
    \begin{array}{cc}
        0 & 2.80\\
        -2.80 & 0
    \end{array}
    \right).
\end{equation}

Similar to the previous example, the algorithm was stopped when the upper bound on the error in~\eqref{eq:theo:Convergence:01} dropped below~$10^{-5}$. In this case, four iterations sufficed for the error to drop below the threshold.

\section{Conclusion}
We considered a GMAC scenario in which the goal is to minimize a general non-linear objective, provided that multiple power constraints are satisfied. The variables that underlie this optimization are the input covariance matrices of the users, their time-sharing parameters and their decoding order. As such, the considered problems fall under the category of mixed-integer optimization  problems, which are generally difficult to solve. To circumvent this difficulty, we invoked Carath\'{e}odory's theorem and variational inequalities to analyze problems with general, possibly non-convex, objectives.  This analysis enabled us  to derive: 1) necessary optimality conditions for general problems; and 2) necessary and sufficient optimality conditions for problems with objectives that are convex in the rates, but not necessarily convex in the time-sharing parameters, decoding orders and covariance matrices. Drawing insight into these conditions, we developed and analyzed the convergence of an algorithm for solving a broad class of practical, but generally nonconvex and difficult to solve, GMAC optimization problems. We suspect that, using the GMAC-GBC duality, our  results can be utilized to determine the optimal transmission parameters of  the dirty-paper coding scheme in various broadcast communication scenarios.

\section*{Acknowledgements}
The authors would like to thank Dr. Melda Yuksel, TOBB University of Economics and Technology, Turkey, for her valuable comments.

\appendices
\section{Background}
\label{preliminaries}
In this appendix, we review background material that involves variational inequalities and that is necessary for the development of some of the results of this paper.

\begin{proposition}\label{theo:TransLinNec}
Let $\opt{\boldsymbol{x}}$ be an optimum solution for
\begin{equation}\label{theo:TransLinNec:01}
\underset{\boldsymbol{x}\in\mathcal{X}}{\min}\ f(\boldsymbol{x}),
\end{equation}
where $\mathcal{X}\subset\mathbb{R}^N$ is a convex set, and $f$ is  continuously differentiable. Then, for the constant vector $\boldsymbol{a}=\nabla f(\opt{\boldsymbol{x}})$, $\opt{\boldsymbol{x}}$ solves the following optimization problem:
\begin{equation}\label{theo:TransLinNec:02}
\underset{\boldsymbol{x}\in\mathcal{X}}{\min}\ \boldsymbol{x}^\intercal\boldsymbol{a}.
\end{equation}
\end{proposition}
\begin{proof}
Let $\nonopt{\boldsymbol{x}}\in\mathcal{X}$ be a point such that ${\opt{\boldsymbol{x}}}^\intercal\nabla f(\opt{\boldsymbol{x}})>{\nonopt{\boldsymbol{x}}}^\intercal\nabla f(\opt{\boldsymbol{x}})$. We are going to show that $\opt{\boldsymbol{x}}$ cannot be a minimum of $f$ in $\mathcal{X}$. Taking the directional derivative of $f$ in the direction $\boldsymbol{v}=\nonopt{\boldsymbol{x}}-\opt{\boldsymbol{x}}$, we have
\begin{equation}
\begin{split}
\nabla_{\boldsymbol{v}}f(\opt{\boldsymbol{x}})&=\lim_{a\rightarrow0^+}\frac{f(\opt{\boldsymbol{x}}+a\boldsymbol{v})-f(\opt{\boldsymbol{x}})}{a|\boldsymbol{v}|}\\
&=\frac{\left(\nonopt{\boldsymbol{x}}-\opt{\boldsymbol{x}}\right)^\intercal}{|\boldsymbol{v}|}\nabla f(\opt{\boldsymbol{x}})<0,
\end{split}
\end{equation}
where the strict inequality follows from the assumption that ${\opt{\boldsymbol{x}}}^\intercal\nabla f(\opt{\boldsymbol{x}})>{\nonopt{\boldsymbol{x}}}^\intercal\nabla f(\opt{\boldsymbol{x}})$. Since $\mathcal{X}$ is convex, $a\nonopt{\boldsymbol{x}}+(1-a)\opt{\boldsymbol{x}}\in\mathcal{X}$, for all $0\leq a\leq1$. Thus, $\opt{\boldsymbol{x}}+a\boldsymbol{v}\in\mathcal{X}$, for all $0\leq a\leq1$. The fact that the previous derivative is strictly negative implies that, for a sufficiently small $a$, the numerator of the limit must be negative. Thus, $f(\opt{\boldsymbol{x}}+a\boldsymbol{v})<f(\opt{\boldsymbol{x}})$, for small $a>0$. Consequently, $\opt{\boldsymbol{x}}$ cannot be a minimum of $f$ in $\mathcal{X}$.
\end{proof}

Proposition~\ref{theo:TransLinNec} describes a necessary condition that is satisfied by all optimum solutions of~\eqref{theo:TransLinNec:01}. In particular, Proposition~\ref{theo:TransLinNec} identifies all stationary points which are guaranteed to be local minima except for those stationary points with null gradients. The latter points can be also local maxima or saddle points that lie within the feasible set $\mathcal{X}$. As such, the global minima constitute a subset of the stationary points identified by Proposition~\ref{theo:TransLinNec}. Complementary to Proposition~\ref{theo:TransLinNec} is the following result.

\begin{proposition}\label{theo:TransLinSuf}
Let $\opt{\boldsymbol{x}}\in\mathcal{X}\subset\mathbb{R}^N$, and let $\boldsymbol{a}=\nabla f(\opt{\boldsymbol{x}})$, where $f$ is convex and continuously differentiable. If $\opt{\boldsymbol{x}}\in\arg\underset{\boldsymbol{x}\in\mathcal{X}}{\min}\ \boldsymbol{x}^\intercal\boldsymbol{a}$, then $\opt{\boldsymbol{x}}$ solves the following optimization problem:
\begin{equation}
\underset{\boldsymbol{x}\in\mathcal{X}}{\min}\ f(\boldsymbol{x}).
\end{equation}
\end{proposition}

\begin{proof}
Let $\opt{\boldsymbol{x}}\in\mathcal{X}$ be a point such that ${\opt{\boldsymbol{x}}}^\intercal\nabla f(\opt{\boldsymbol{x}})\leq\boldsymbol{x}^\intercal\nabla f(\opt{\boldsymbol{x}})$ for all $\boldsymbol{x}\in\mathcal{X}$. Then, $0\leq\left(\boldsymbol{x}-\opt{\boldsymbol{x}}\right)^\intercal\nabla f(\opt{\boldsymbol{x}})$. Since $f$ is convex, then, from the first order convexity condition~\cite[Proposition~B.3]{Bertsekas}, $f(\boldsymbol{x})-f(\opt{\boldsymbol{x}})\geq\left(\boldsymbol{x}-\opt{\boldsymbol{x}}\right)^\intercal\nabla f(\opt{\boldsymbol{x}})$. Therefore, $f(\boldsymbol{x})-f(\opt{\boldsymbol{x}})\geq0$, and, hence, $f(\boldsymbol{x})\geq f(\opt{\boldsymbol{x}})$.
\end{proof}

The results of Propositions~\ref{theo:TransLinNec} and~\ref{theo:TransLinSuf} are usually presented conjointly. However, for the purpose of this paper, it is essential to make the distinction between these results. In particular, the necessary optimality condition in Proposition~\ref{theo:TransLinNec} follows from the convexity of the feasible set, and the sufficient optimality condition in Proposition~\ref{theo:TransLinSuf} follows from the convexity of the objective function.

\section{Proof of Proposition~\ref{theo:GMACLinearNS}}
\label{App:Proof:GMACLinear}
\label{App:Proof:GMACLinearNS}

The proof of this proposition is composed of two steps. In the first step, we will show the optimality  of the orderings that follow the increasing order of the weights. In the second step, we will show that, if another ordering is also optimum, then the rates achieved by this decoding order are equal to the rates achieved by an order that follows the increasing order of the weights.

\paragraph*{Step 1}
Let $i$ and $\opt{i}$ be two identical decoding orderings except for positions $k$ and $k+1$, which are reversed. Without loss of generality, let the users be labeled following the ordering $\opt{i}$. Then, in the ordering $i$, user $k$ is decoded right after user $k+1$, and in the ordering $\opt{i}$ user $k$ is decoded right before user $k+1$; and we have that
\begin{equation}\label{eq:theo:GMACLinear:02}
k\!=\!\pi_{\opt{i}}^{-1}(k)\!=\!\pi_{\opt{i}}^{-1}(k+1)-1\!=\!\pi_{i}^{-1}(k)-1\!=\!\pi_{i}^{-1}(k+1).
\end{equation}
Suppose that $w_k\leq w_{k+1}$. We will show that the $\opt{i}$-th ordering is always better than the $i$-th ordering, which implies that users $k$ and $k+1$ should be decoded following the increasing order of their weights.

Let $\cccov=\{\ccov[(m)]\}_{m=1}^{K+1}$ and $m\in\{1,\dots,K+1\}$ be given. Then, from~\eqref{eq:r_ki} and~\eqref{eq:theo:GMACLinear:02}, the sum-rate of users $k$ and $k+1$ in the $\opt{i}$-th ordering is
\begin{equation}\label{eq:theo:GMACLinear:03b}
\begin{split}
r_{k\opt{i}}&\big(\ccov[(m)]\big)+r_{(k+1)\opt{i}}\big(\ccov[(m)]\big)=\\
&\log\frac{\left|\boldsymbol{I}\!+\!\sum_{j\geq\pi_{\opt{i}}^{-1}(k)}\boldsymbol{H}_{\pi_{\opt{i}}(j)}\ucov[(m)][\pi_{\opt{i}}(j)]\boldsymbol{H}_{\pi_{\opt{i}}(j)}^\dagger\right|}{\left|\boldsymbol{I}\!+\!\sum_{j>\pi_{\opt{i}}^{-1}(k)}\boldsymbol{H}_{\pi_{\opt{i}}(j)}\ucov[(m)][\pi_{\opt{i}}(j)]\boldsymbol{H}_{\pi_{\opt{i}}(j)}^\dagger\right|}+\\
&\log\frac{\left|\boldsymbol{I}\!+\!\sum_{j\geq\pi_{\opt{i}}^{-1}(k+1)}\boldsymbol{H}_{\pi_{\opt{i}}(j)}\ucov[(m)][\pi_{\opt{i}}(j)]\boldsymbol{H}_{\pi_{\opt{i}}(j)}^\dagger\right|}{\left|\boldsymbol{I}\!+\!\sum_{j>\pi_{\opt{i}}^{-1}(k+1)}\boldsymbol{H}_{\pi_{\opt{i}}(j)}\ucov[(m)][\pi_{\opt{i}}(j)]\boldsymbol{H}_{\pi_{\opt{i}}(j)}^\dagger\right|}=\\
&\log\frac{\left|\boldsymbol{I}\!+\!\sum_{j\geq k}\boldsymbol{H}_{\pi_{\opt{i}}(j)}\ucov[(m)][\pi_{\opt{i}}(j)]\boldsymbol{H}_{\pi_{\opt{i}}(j)}^\dagger\right|}{\left|\boldsymbol{I}\!+\!\sum_{j>k}\boldsymbol{H}_{\pi_{\opt{i}}(j)}\ucov[(m)][\pi_{\opt{i}}(j)]\boldsymbol{H}_{\pi_{\opt{i}}(j)}^\dagger\right|}+\\
&\log\frac{\left|\boldsymbol{I}\!+\!\sum_{j>k}\boldsymbol{H}_{\pi_{\opt{i}}(j)}\ucov[(m)][\pi_{\opt{i}}(j)]\boldsymbol{H}_{\pi_{\opt{i}}(j)}^\dagger\right|}{\left|\boldsymbol{I}\!+\!\sum_{j>k+1}\boldsymbol{H}_{\pi_{\opt{i}}(j)}\ucov[(m)][\pi_{\opt{i}}(j)]\boldsymbol{H}_{\pi_{\opt{i}}(j)}^\dagger\right|}=\\
&\log\frac{\left|\boldsymbol{I}\!+\!\sum_{j\geq k}\boldsymbol{H}_{\pi_{\opt{i}}(j)}\ucov[(m)][\pi_{\opt{i}}(j)]\boldsymbol{H}_{\pi_{\opt{i}}(j)}^\dagger\right|}{\left|\boldsymbol{I}\!+\!\sum_{j>k+1}\boldsymbol{H}_{\pi_{\opt{i}}(j)}\ucov[(m)][\pi_{\opt{i}}(j)]\boldsymbol{H}_{\pi_{\opt{i}}(j)}^\dagger\right|}.
\end{split}
\end{equation}
Similarly, for the $i$-th ordering we have
\begin{multline}\label{eq:theo:GMACLinear:03c}
r_{ki}\big(\ccov[(m)]\big)+r_{(k+1)i}\big(\ccov[(m)]\big)=\\
\log\frac{\left|\boldsymbol{I}\!+\!\sum_{j\geq k}\boldsymbol{H}_{\pi_{\opt{i}}(j)}\ucov[(m)][\pi_{\opt{i}}(j)]\boldsymbol{H}_{\pi_{\opt{i}}(j)}^\dagger\right|}{\left|\boldsymbol{I}\!+\!\sum_{j>k+1}\boldsymbol{H}_{\pi_{\opt{i}}(j)}\ucov[(m)][\pi_{\opt{i}}(j)]\boldsymbol{H}_{\pi_{\opt{i}}(j)}^\dagger\right|}.
\end{multline}
Hence, from~\eqref{eq:theo:GMACLinear:03b} and~\eqref{eq:theo:GMACLinear:03c}, we have that,
\begin{multline}\label{eq:theo:GMACLinear:04}
r_{ki}\big(\ccov[(m)]\big)+r_{(k+1)i}\big(\ccov[(m)]\big)=\\
r_{k\opt{i}}\big(\ccov[(m)]\big)+r_{(k+1)\opt{i}}\big(\ccov[(m)]\big).
\end{multline}

Since the $k$-th user sees the signal of user $k+1$ as interference in the $\opt{i}$-th ordering, but not in the $i$-th ordering, then, it can be seen from~\eqref{eq:r_ki} and the relation~\cite{BC_Capacity_Region_J}
\begin{equation}\label{eq:theo:OneCovMatFixed_Prop:ineq}
\frac{|\boldsymbol{I}+\boldsymbol{A}+\boldsymbol{B}+\boldsymbol{C}|}{|\boldsymbol{I}+\boldsymbol{A}+\boldsymbol{B}|}\leq\frac{|\boldsymbol{I}+\boldsymbol{A}+\boldsymbol{C}|}{|\boldsymbol{I}+\boldsymbol{A}|},
\end{equation}
which holds for all positive semidefinite matrices $\boldsymbol{A}$, $\boldsymbol{B}$, and $\boldsymbol{C}$, that
\begin{equation}\label{eq:theo:GMACLinear:04b}
r_{ki}\big(\ccov[(m)]\big)\geq r_{k\opt{i}}\big(\ccov[(m)]\big).
\end{equation}
We now multiply~\eqref{eq:theo:GMACLinear:04} with $w_{k+1}$, and multiply~\eqref{eq:theo:GMACLinear:04b} with $w_k-w_{k+1}$. Adding the left hand sides and the right hand sides of the resulting inequalities, yields
\begin{multline}\label{eq:theo:GMACLinear:06}
w_kr_{ki}\big(\ccov[(m)]\big)+w_{k+1}r_{(k+1)i}\big(\ccov[(m)]\big)\leq\\
w_kr_{k\opt{i}}\big(\ccov[(m)]\big)+w_{k+1}r_{(k+1)\opt{i}}\big(\ccov[(m)]\big),
\end{multline}
which shows that the $\opt{i}$-th ordering yields a higher objective than its $i$-th counterpart.

\paragraph*{Step 2}
Let $w_k$ be strictly less than $w_{k+1}$, i.e., $w_k<w_{k+1}$. Then, using the decoding order $i$, users $k$ and $k+1$ are decoded following the decreasing order of their weights. From~\eqref{eq:theo:GMACLinear:06},
it can be seen that the $i$-th ordering yields an objective that is at most as great as that of the $\opt{i}$-th ordering. In this case, \eqref{eq:theo:GMACLinear:06} is satisfied with equality, and, subsequently, \eqref{eq:theo:GMACLinear:04b} is also satisfied with equality, i.e., $r_{ki}\big(\ccov[(m)]\big)=r_{k\opt{i}}\big(\ccov[(m)]\big)$.
Moreover, from~\eqref{eq:theo:GMACLinear:04}, the sum of the rates of both users is the same in the two orderings. Therefore, $r_{(k+1)i}\big(\ccov[(m)]\big)=r_{(k+1)\opt{i}}\big(\ccov[(m)]\big)$.
As a consequence, the rates must be the same in orderings $\opt{i}$ and $i$.

\section{Proof of Proposition~\ref{theo:OneCovMatLinObjA}}
\label{App:Proof:OneCovMatLinObjA}
The proof begins by assuming that an optimum pair $(\opt{\boldsymbol{\alpha}},\cccovo)$ uses two or more composite covariance matrices, and proceeds to show that one of these matrices suffices to maximize the linear rate objective.

Let $\opt{\boldsymbol{\alpha}}$ and $\cccovo$ be optimum for the problem in~\eqref{eq:MACPLinear}. Assume that $\opt{\alpha}_{mi}=0$ for all $m$ and all $i\neq\opt{i}$, that is to say, only the $\opt{i}$-th ordering is used in $\opt{\boldsymbol{\alpha}}$.
In this case, the value of the objective in~\eqref{eq:MACPLinear} for this $\opt{\boldsymbol{\alpha}}$ and $\cccovo$ is given by
\begin{equation}\label{eq:theo:OneCovMatLinObjA:01}
\sum_{k=1}^Kw_k\rho_k(\opt{\boldsymbol{\alpha}},\cccovo)=\sum_{k=1}^K\sum_{m=1}^{K+1}w_k\opt{\alpha}_{m\opt{i}}r_{k\opt{i}}\big(\ccovo[(m)]\big).
\end{equation}

Suppose that the optimum is achieved with at least two composite covariance matrices, and let $\mathcal{M}\subseteq\{1,\dots,K+1\}$ be the set of indices of all the used composite covariance matrices. By definition of the time-sharing matrix $\boldsymbol{\alpha}$ (cf. Section~\ref{model}), $m\in\mathcal{M}$ if and only if $\opt{\alpha}_{m\opt{i}}>0$. If, for some $m,n\in\mathcal{M}$, $\sum_{k=1}^Kw_kr_{k\opt{i}}\big(\ccovo[(m)]\big)>\sum_{k=1}^Kw_kr_{k\opt{i}}\big(\ccovo[(n)]\big)$, then, time-sharing of $\ccovo[(m)]$ and $\ccovo[(n)]$ cannot be optimum, since it would be better to use $\ccovo[(m)]$ alone, and not with $\ccovo[(n)]$. Therefore, for all $m,n\in\mathcal{M}$, we must have
$\sum_{k=1}^Kw_kr_{k\opt{i}}\big(\ccovo[(m)]\big)=\sum_{k=1}^Kw_kr_{k\opt{i}}\big(\ccovo[(n)]\big)=C$, where $C$ is a constant. In this case, the equality in~\eqref{eq:theo:OneCovMatLinObjA:01} can be expressed as
\begin{equation}\label{eq:theo:OneCovMatLinObjA:02}
\sum_{k=1}^Kw_k\rho_k(\opt{\boldsymbol{\alpha}},\cccovo)=\sum_{m=1}^{K+1}\opt{\alpha}_{m\opt{i}}C=C,
\end{equation}
where the last equality follows from the fact that $\sum_{m=1}^{K+1}\opt{\alpha}_{m\opt{i}}=1$. Note that~\eqref{eq:theo:OneCovMatLinObjA:02} is independent of the values of $\opt{\alpha}_{m\opt{i}}$, $m=1,\dots,K+1$. Hence, the value of the objective is independent of the time-sharing of the composite covariance matrices with indices in $\mathcal{M}$. In other words, we can choose $\opt{\alpha}_{m\opt{i}}=1$ for some $m\in\mathcal{M}$, and $\opt{\alpha}_{n\opt{i}}=0$ for every $n\neq m$, which implies that one composite covariance matrix suffices to achieve the maximum of the linear rate objective.

\section{Proof of Theorem~\ref{theo:LinNecc}}
\label{App:Proof:theo:LinNecc}

\subsection{Proof of the direct part}
In this section we will show that if the pair $(\opt{\boldsymbol{\alpha}},\cccovo)$ solves~\eqref{eq:MACP}, then it satisfies the two conditions of this theorem. To do so, we will synthesize an auxiliary optimization problem with a linear rate objective such that a solution for the problem in~\eqref{eq:MACP} is also a solution of the auxiliary problem. The two conditions of this theorem are then obtained by invoking the results of Section~\ref{OptGMAC:Case02}.

The convexity of the set of all achievable rates, $\mathcal{C}$, implies that Proposition~\ref{theo:TransLinNec} can be applied to solve the following optimization problem:
\begin{equation}\label{eq:theo:LinNecc:02}
\underset{\boldsymbol{x}\in\mathcal{C}}{\min}\ f(\boldsymbol{x}).
\end{equation}
Since $\opt{\boldsymbol{\alpha}}$ and $\cccovo$ are optimum, it follows that
$\opt{\boldsymbol{x}}=\boldsymbol{\rho}\left(\opt{\boldsymbol{\alpha}},\cccovo\right)$
solves~\eqref{eq:theo:LinNecc:02}. However, from Proposition~\ref{theo:TransLinNec}, $\opt{\boldsymbol{x}}\in\arg\max_{\boldsymbol{x}\in\mathcal{C}}\ \sum_{k=1}^Kw_kx_k$, where the weights are defined as in~\eqref{eq:theo:LinNecc:00}. Hence, $\opt{\boldsymbol{\alpha}}$ and $\cccovo$ solve
\begin{equation}\label{eq:theo:LinNecc:03}
    \underset{\boldsymbol{\alpha}\in\mathcal{S}_\alpha,\;\ccov[(m)]\in\feascov,\;m=1,\dots,K+1}{\max}\quad\sum_{k=1}^Kw_k\rho_k\left(\boldsymbol{\alpha},\cccov\right).
\end{equation}
Since the objective of this problem is linear in the rates, the optimum orderings must satisfy Proposition~\ref{theo:GMACLinearNS} in Section~\ref{OptGMAC:Case02}. This yields the first condition of the theorem.

To prove the second condition, we begin by obtaining an expression for the time-sharing matrix that solves~\eqref{eq:theo:LinNecc:03} for an arbitrary collection of composite covariance matrices. Substituting this expression in~\eqref{eq:theo:LinNecc:03} yields a problem in which the only optimization variables are the covariance matrices. Finally, we will show that the resulting problem is equivalent to the one in~\eqref{eq:theo:LinNecc:01}.

To determine the optimum time-sharing matrix, we assume that the covariance matrices are fixed. In this case, the problem in~\eqref{eq:theo:LinNecc:03} becomes a linear program in the time-sharing matrix, which, using the definition of $\rho_k(\boldsymbol{\alpha},\cccov)$ in~\eqref{eq:rate_k}, can be cast as
\begin{equation}\label{eq:theo:LinNecc:04}
    \underset{\boldsymbol{\alpha}\in\mathcal{S}_\alpha}{\max}\quad\sum_{k=1}^K\sum_{m=1}^{K+1}\sum_{i=1}^{K!}\alpha_{mi}w_kr_{ki}\big(\ccov[(m)]\big).
\end{equation}
It is straightforward to verify that the solution of this problem satisfies
\begin{subequations}\label{eq:theo:LinNecc:05}
\begin{align}
\alpha_{mi}=0,&\text{ if }\sum_{k=1}^Kw_kr_{ki}\big(\ccov[(m)]\big)<B,\text{ and}\label{eq:theo:LinNecc:05:01}\\
\alpha_{mi}\geq0,&\text{ if }\sum_{k=1}^Kw_kr_{ki}\big(\ccov[(m)]\big)=B,\label{eq:theo:LinNecc:05:02}%
\end{align}
\end{subequations}
where $B=\max_{mi}\sum_{k=1}^Kw_kr_{ki}\big(\ccov[(m)]\big)$.

Substituting for the set $\{\alpha_{mi}\}$ from~\eqref{eq:theo:LinNecc:05} into~\eqref{eq:theo:LinNecc:03}, we can write an equivalent problem that does not involve the time-sharing matrix. In particular, since $\boldsymbol{\alpha}$ satisfies~\eqref{eq:theo:LinNecc:05}, the objective in~\eqref{eq:theo:LinNecc:03} can be expressed as $\sum_{k=1}^Kw_k\rho_k\left(\boldsymbol{\alpha},\cccov\right)=\sum_{m=1}^{K+1}\sum_{i=1}^{K!}\alpha_{mi}B=B$, where the second equality follows from the fact that $\sum_{m=1}^{K+1}\sum_{i=1}^K\alpha_{mi}=1$. Furthermore, since $\cccovo$ is optimum for~\eqref{eq:theo:LinNecc:03}, it is also optimum for
\begin{equation}\label{eq:theo:LinNecc:06}
    \underset{\ccov[(m)]\in\feascov,\;m=1,\dots,K+1}{\max}\quad\max_{mi}\sum_{k=1}^Kw_kr_{ki}\big(\ccov[(m)]\big).
\end{equation}
Now, since $\opt{\boldsymbol{\alpha}}$ satisfies~\eqref{eq:theo:LinNecc:05}, the pairs $\{(m,i)\}$ that maximize~\eqref{eq:theo:LinNecc:06}\footnote{The optimum solution of~\eqref{eq:MACP} might use time-sharing between different composite covariance matrices and decoding orders. In that case, multiple pairs $\{(m,i)\}$ maximize~\eqref{eq:theo:LinNecc:06}.} are those for which $\opt{\alpha}_{mi}>0$. In that case, $\ccovo[(m)]$ is optimum for
\begin{equation}\label{eq:theo:LinNecc:07}
    \underset{\ccov[]\in\feascov}{\max}\quad\sum_{k=1}^Kw_kr_{ki}\big(\ccov[]\big).
\end{equation}

The problem in~\eqref{eq:theo:LinNecc:07} does not contain the time-sharing matrix. Moreover, as we have shown, the $m$-th optimum composite covariance matrix of~\eqref{eq:theo:LinNecc:03}, $\ccovo[(m)]$, is also optimum for~\eqref{eq:theo:LinNecc:07} for the ordering satisfying $\opt{\alpha}_{mi}>0$. Our next step is to show that the problem in~\eqref{eq:theo:LinNecc:07} is equivalent to that in~\eqref{eq:theo:LinNecc:01}.

Let $\opt{i}$ be the increasing order of the weights and let $w_1\leq\cdots\leq w_K$. From Proposition~\ref{theo:GMACLinearNS}, it can be seen that $\opt{i}$ is optimum. However, if $\opt{\alpha}_{mi}>0$, the $i$-th ordering is also optimum. In that case, the optimization problem in~\eqref{eq:theo:LinNecc:07} is equivalent to
$\max_{\ccov[]\in\feascov}\ \sum_{k=1}^Kw_kr_{k\opt{i}}\big(\ccov[]\big)$,
which, using~\eqref{eq:prop:02}, is equivalent to $\max_{\ccov[]\in\feascov}\ \sum_{k=1}^K(w_k-w_{k-1})\log\big|\boldsymbol{I}\!+\!\sum_{j\geq
k}\boldsymbol{H}_j\ucov[][j]\boldsymbol{H}_j^\dagger\big|$, which is identical to~\eqref{eq:theo:LinNecc:01}.

\subsection{Proof of the converse}
In this section we will show that if the function $f$ is convex in the rates and the pair $(\opt{\boldsymbol{\alpha}},\cccovo)$ satisfies the two conditions of this theorem, then it solves~\eqref{eq:MACP}. To do so, we will proceed in two steps. First, we will show that a pair $(\opt{\boldsymbol{\alpha}},\cccovo)$ that satisfies the two conditions of this theorem solves~\eqref{eq:theo:LinNecc:03}. Second, we will use the convexity of $f$ and Proposition~\ref{theo:TransLinSuf} to show that a pair that solves~\eqref{eq:theo:LinNecc:03} must solve~\eqref{eq:MACP}.

\paragraph*{Step 1}
Let $\opt{\boldsymbol{\alpha}}$ and $\cccovo$ satisfy the two conditions of this theorem, and let the $mi$-th element of $\opt{\boldsymbol{\alpha}}$ be strictly positive, that is, $\opt{\alpha}_{mi}>0$. Then, from the conditions of this theorem, the \mbox{$i$-th} ordering satisfies Proposition~\ref{theo:GMACLinearNS}, and $\ccovo[(m)]$ is optimum for (\ref{eq:theo:LinNecc:01}). Recalling the equivalence of the optimization problems in~\eqref{eq:theo:LinNecc:07} and~\eqref{eq:theo:LinNecc:01}, and the optimality of the $i$-th ordering for the weights $\{w_k\}_{k=1}^K$, we have that, for any feasible collection of composite covariance matrices $\cccovno=\{\ccovno[(n)]\}_{n=1}^{K+1}$,
\begin{equation}\label{eq:theo:LinSuff:04a}
\sum_{k=1}^Kw_kr_{ki}\big(\ccovo[(m)]\big)\!\geq\!\sum_{k=1}^Kw_kr_{ki}\big(\ccovno[(n)]\big)\!\geq\!\sum_{k=1}^Kw_kr_{kj}\big(\ccovno[(n)]\big),
\end{equation}
for any composite covariance matrix $n\in\{1,\dots,K+1\}$, and ordering $j\in\{1,\dots,K!\}$. Let $\nonopt{\boldsymbol{\alpha}}$ be any feasible time-sharing matrix satisfying $\sum_{n=1}^{K+1}\sum_{j=1}^{K!}\nonopt{\alpha}_{nj}=1$. Then, the inequality in~\eqref{eq:theo:LinSuff:04a} can be written as
\begin{equation}\label{eq:theo:LinSuff:05}
\begin{split}
    \sum_{k=1}^Kw_kr_{ki}\big(\ccovo[(m)]\big)&\geq\sum_{k=1}^K\sum_{n=1}^{K+1}\sum_{j=1}^{K!}w_k\nonopt{\alpha}_{nj}r_{kj}\big(\ccovno[(n)]\big)\\
    &=\sum_{k=1}^Kw_k\rho_k\big(\nonopt{\boldsymbol{\alpha}},\cccovno\big),
\end{split}
\end{equation}
where the equality follows from~\eqref{eq:rate_k}. Since~\eqref{eq:theo:LinSuff:05} holds for any feasible $\nonopt{\boldsymbol{\alpha}}$ and $\cccovno$, it follows that $\sum_{k=1}^Kw_kr_{ki}\big(\ccovo[(m)]\big)$
is an upper bound on the objective of~\eqref{eq:theo:LinNecc:03}. Since this holds for all pairs $(m,i)$ that satisfy $\opt{\alpha}_{mi}>0$, then $\opt{\boldsymbol{\alpha}}$ and $\cccovo$ solve~\eqref{eq:theo:LinNecc:03}.

\paragraph*{Step 2}
It remains to show that a pair $(\opt{\boldsymbol{\alpha}},\cccovo)$ that solves~\eqref{eq:theo:LinNecc:03} must solve~\eqref{eq:MACP}. The optimality of $(\opt{\boldsymbol{\alpha}},\cccovo)$ for the problem in~\eqref{eq:theo:LinNecc:03} implies that $\opt{\boldsymbol{x}}=\boldsymbol{\rho}\left(\opt{\boldsymbol{\alpha}},\cccovo\right)$ is optimum for $\max_{\boldsymbol{x}\in\mathcal{C}}\ \sum_{k=1}^Kw_kx_k$.
From~\eqref{eq:theo:LinNecc:00} and Proposition~\ref{theo:TransLinSuf}, $\opt{\boldsymbol{x}}$ is also optimum for $\min_{\boldsymbol{x}\in\mathcal{C}}\ f(\boldsymbol{x})$, which implies that $\opt{\boldsymbol{\alpha}}$ and $\cccovo$ are optimum for~\eqref{eq:MACP}.

\section{Proof of Proposition~\ref{theo:Alg:Stepsize}}
\label{App:Proof:Alg:Stepsize}
From~\eqref{eq:Alg:cov} and~\eqref{eq:Alg:h}, the inequality~\eqref{eq:theo:Alg:f_decrease} can be expressed as $h(\varepsilon(t),t)\leq h(0,t)$. Using~\eqref{eq:Alg:epsilon}, the proposition is proved by showing that $\frac{\partial h(\delta,t)}{\partial\delta}\big|_{\delta=0}\leq0$. Towards that end, we define the function
\begin{equation}\label{eq:theo:Convergence:02}
\hat{h}(\delta,t)=\sum_{k=1}^K\!w_k(t)\hat\rho_k\big(\opt{\boldsymbol{\beta}}(t-1),\delta\ccovno(t)+(1-\delta)\ccovo(t-1)\big).
\end{equation}
Now,
\begin{multline}\label{eq:theo:Convergence:06b}
\frac{\partial \hat{h}(\delta,t)}{\partial\delta}\bigg|_{\delta=0}=
\sum_{k=1}^Kw_k(t)\\
\frac{\partial\hat\rho_k\big(\opt{\boldsymbol{\beta}}(t-1),\delta\ccovno(t)+(1-\delta)\ccovo(t-1)\big)}{\partial\delta}\bigg|_{\delta=0},
\end{multline}
and
\begin{multline}\label{eq:theo:Convergence:06c}
\frac{\partial h(\delta,t)}{\partial\delta}\bigg|_{\delta=0}=
\sum_{k=1}^K\frac{\partial f(\boldsymbol{x})}{\partial x_k}\bigg|_{\boldsymbol{x}=\hat{\boldsymbol{\rho}}\left(\opt{\boldsymbol{\beta}}(t-1),\ccovo(t-1)\right)}\\
\frac{\partial\hat\rho_k\big(\opt{\boldsymbol{\beta}}(t-1),\delta\ccovno(t)+(1-\delta)\ccovo(t-1)\big)}{\partial\delta}\bigg|_{\delta=0}.
\end{multline}
Hence, recalling that (cf. Algorithm~\ref{Alg:Algorithm})
\begin{equation}\label{eq:theo:Convergence:06d}
    w_k(t)=-\frac{\partial f(\boldsymbol{x})}{\partial x_k}\bigg|_{\boldsymbol{x}=\hat{\boldsymbol{\rho}}\left(\opt{\boldsymbol{\beta}}(t-1),\ccovo(t-1)\right)},\ k=1,\dots,K,
\end{equation}
we conclude that
\begin{equation}\label{eq:theo:Convergence:07}
\frac{\partial h(\delta,t)}{\partial\delta}\Big|_{\delta=0}=-\frac{\partial \hat{h}(\delta,t)}{\partial\delta}\Big|_{\delta=0}.
\end{equation}

To show that $\frac{\partial \hat h(\delta,t)}{\partial\delta}\big|_{\delta=0}\geq0$, we recall that, in the sixth step of Algorithm~\ref{Alg:Algorithm}, $\opt{\boldsymbol{\beta}}(t-1)$ solves ${\min}_{\boldsymbol{\beta}\in\mathcal{S}_\beta}\ f\big(\hat{\boldsymbol{\rho}}\big(\boldsymbol{\beta},\ccovo[](t-1)\big)\big)$, which is a special case of the problem in~\eqref{eq:MACP2}. As such, any ordering $i$, such that $\opt{\beta}_i(t-1)>0$, is optimum for the linear objective defined by the weights $\{w_k(t)\}_{k=1}^K$, which implies that the weighted sum rate for all these orderings is equal to $\hat h(\delta,t)$, i.e.,
\begin{equation}\label{eq:theo:Convergence:02b}
\sum_{k=1}^Kw_k(t)r_{ki}\big(\delta\ccovno(t)+(1-\delta)\ccovo(t-1)\big)=\hat h(\delta,t),
\end{equation}
for all $i$ such that $\opt{\beta}_i(t-1)>0$, where the last equality follows from~\eqref{eq:theo:Convergence:02}, the definition of $\hat{\rho}_k(\boldsymbol{\beta},\ccov)$ in~\eqref{eq:RVC}, and the fact that $\sum_{i=1}^{K!}\opt{\beta}_i(t-1)=1$. Assuming that $w_1(t)\leq\cdots\leq w_K(t)$ and using~\eqref{eq:theo:Convergence:02b} and~\eqref{eq:prop:02}, $\hat h(\delta,t)$ can be expressed as
\begin{multline}\label{eq:theo:Convergence:03}
\hat h(\delta,t)=\sum_{k=1}^K\big(w_k(t)-w_{k-1}(t)\big)\\
\log\Bigl|\boldsymbol{I}+\sum_{j\geq k}\boldsymbol{H}_j\big(\delta\ucovno[][j](t)+(1-\delta)\ucovo[][j](t-1)\big)\boldsymbol{H}_j^\dagger\Bigr|,
\end{multline}
where $w_0(t)=0$. The function $\hat h(\delta,t)$ can be shown to be concave in $\delta$. Since $\ccovno(t)$ solves~\eqref{eq:theo:LinNecc:01} for the weights $\{w_k(t)\}_{k=1}^K$, setting $\delta=1$ solves $\max_{\delta\in[0,1]}\hat h(\delta,t)$. Now, the optimality of $\delta=1$, the concavity of $\hat h$, and the result in~\cite[Proposition~B.3]{Bertsekas} yield
\begin{equation}\label{eq:theo:Convergence:03b}
\frac{\partial \hat h(\delta,t)}{\partial\delta}\bigg|_{\delta=0}\geq \hat h(1,t)-\hat h(0,t)\geq0,
\end{equation}
which, by~\eqref{eq:theo:Convergence:07}, completes the proof.

\section{Proof of Proposition~\ref{theo:Alg:Error}}
\label{App:Proof:Alg:Error}
Since $f$ is convex and continuously differentiable, we have that~\cite[Proposition~B.3]{Bertsekas}
\begin{multline}
    f(\opt{\boldsymbol{x}})-f\big(\hat{\boldsymbol{\rho}}\big(\opt{\boldsymbol{\beta}}\!(t),\ccovo\!(t)\big)\big)\geq\\
    \sum_{k=1}^K\big(\opt{x}_k\!-\!\hat\rho_k\big(\opt{\boldsymbol{\beta}}\!(t),\ccovo\!(t)\big)\big)\frac{\partial f(\boldsymbol{x})}{\partial x_k}\bigg|_{\boldsymbol{x}=\hat{\boldsymbol{\rho}}\left(\opt{\boldsymbol{\beta}}\!(t),\ccovo\!(t)\right)}.
\end{multline}
Hence, recalling that the weights $\{w_k(t+1)\}_{k=1}^K$ are computed as in~\eqref{eq:theo:LinNecc:00} at the rate vector $\boldsymbol{x}=\hat{\boldsymbol{\rho}}\big(\opt{\boldsymbol{\beta}}(t),\ccovo(t)\big)$, we conclude that
\begin{multline}\label{eq:Algorithm:eq01}
    f(\opt{\boldsymbol{x}})-f\big(\hat{\boldsymbol{\rho}}\big(\opt{\boldsymbol{\beta}}(t),\ccovo(t)\big)\big)\geq\\
    \sum_{k=1}^K\big(\hat\rho_k\big(\opt{\boldsymbol{\beta}}(t),\ccovo(t)\big)-\opt{x}_k\big)w_k(t+1).
\end{multline}
Since $\ccovno(t+1)$ is the composite covariance matrix that solves~\eqref{eq:theo:LinNecc:01} for the weights $\{w_k(t+1)\}_{k=1}^K$ (cf. step 4 of Algorithm~\ref{Alg:Algorithm}), and since the orderings used in $\opt{\boldsymbol{\beta}}(t)$ follow the increasing order of these weights, it follows that
\begin{equation}\label{eq:Algorithm:eq01b}
    \sum_{k=1}^Kw_k(t+1)\opt{x}_k\leq\sum_{k=1}^Kw_k(t+1)\hat\rho_k\left(\opt{\boldsymbol{\beta}}(t),\ccovno(t+1)\right).
\end{equation}
Note that the right hand of~\eqref{eq:Algorithm:eq01b} is the maximum of the objective in~\eqref{eq:theo:LinNecc:01} at iteration $t+1$. Substituting~\eqref{eq:Algorithm:eq01b} into the right hand side of~\eqref{eq:Algorithm:eq01} yields the first inequality of~\eqref{eq:Algorithm:eq02}.

\section{Proof of Theorem~\ref{theo:Convergence}}
\label{App:Proof:Convergence}
To prove this theorem, we begin by noting that the convexity of $f$ implies that $f(\opt{\boldsymbol{x}})\geq f(\boldsymbol{x})+(\opt{\boldsymbol{x}}-\boldsymbol{x})^\intercal\nabla f(\boldsymbol{x})$ for all $\boldsymbol{x}$. Since $f$ is continuously differentiable and the capacity region is bounded, then the norms of $\opt{\boldsymbol{x}}$, $\boldsymbol{x}$, and $\nabla f(\boldsymbol{x})$ are finite. This implies that $f(\opt{\boldsymbol{x}})$ is bounded below if $f(\boldsymbol{x})>-\infty$ for at least one $\boldsymbol{x}$. From~\eqref{eq:theo:Alg:f_decrease2} and~\eqref{eq:theo:Alg:f_decrease}, and the boundedness of $f$, Algorithm~\ref{Alg:Algorithm} must converge to a point. Our goal now is to show that this point is optimal.

We will use Proposition~\ref{theo:Alg:Error} to upper bound the difference between the optimum and the value of the objective at subsequent iterations. In particular, from~\eqref{eq:Algorithm:eq02} we have that
\begin{multline}\label{eq:theo:Convergence:01}
\sum_{k=1}^K\!w_k(t+1)\!\big(\hat\rho_k\!\big(\opt{\boldsymbol{\beta}}(t),\ccovno(t+1)\big)\!-\!\hat\rho_k\!\big(\opt{\boldsymbol{\beta}}(t),\ccovo(t)\big)\big)\geq\\
f\big(\hat{\boldsymbol{\rho}}\big(\opt{\boldsymbol{\beta}}(t),\ccovo(t)\big)\big)-f(\opt{\boldsymbol{x}})\geq0.
\end{multline}
We will prove the theorem by showing that the left hand side of~\eqref{eq:theo:Convergence:01} goes to zero as $t$ goes to $\infty$. Using the definition of $\hat h(\delta,t)$ in~\eqref{eq:theo:Convergence:02} in Appendix~\ref{App:Proof:Alg:Stepsize}, the left hand side of~\eqref{eq:theo:Convergence:01} can be expressed as $\hat h(1,t+1)-\hat h(0,t+1)$. From~\eqref{eq:theo:Convergence:03b}, it is straightforward to show that $\lim_{t\rightarrow\infty}\hat h(1,t)-\hat h(0,t)=0$ if $\lim_{t\rightarrow\infty}\frac{\partial \hat h(\delta,t)}{\partial\delta}\big|_{\delta=0}=0$. Using~\eqref{eq:theo:Convergence:07}, we can prove that $\lim_{t\rightarrow\infty}\frac{\partial \hat h(\delta,t)}{\partial\delta}\big|_{\delta=0}=0$ by showing that $\lim_{t\rightarrow\infty}\frac{\partial h(\delta,t)}{\partial\delta}\big|_{\delta=0}=0$.

To show that $\lim_{t\rightarrow\infty}\frac{\partial h(\delta,t)}{\partial\delta}\big|_{\delta=0}=0$, we begin by noting that
\begin{equation}\label{eq:theo:Convergence:04}
\lim_{t\rightarrow\infty}f\big(\hat{\boldsymbol{\rho}}\big(\opt{\boldsymbol{\beta}}(t-1),\ccovo(t-1)\big)\big)-f\big(\hat{\boldsymbol{\rho}}\big(\opt{\boldsymbol{\beta}}(t),\ccovo(t)\big)\big)=0.
\end{equation}
This, along with~\eqref{eq:theo:Alg:f_decrease2} and~\eqref{eq:theo:Alg:f_decrease}, implies that
\begin{equation}\label{eq:theo:Convergence:05}
\lim_{t\rightarrow\infty}f\big(\hat{\boldsymbol{\rho}}\big(\opt{\boldsymbol{\beta}}(t-1),\ccovo(t-1)\big)\big)-f\big(\hat{\boldsymbol{\rho}}\big(\opt{\boldsymbol{\beta}}(t-1),\ccovo(t)\big)\big)=0.
\end{equation}
Since $\ccovo(t)=\varepsilon(t)\ccovno(t)+(1-\varepsilon(t))\ccovo(t-1)$, cf.~\eqref{eq:Alg:cov}, then, using the definition of $h(\delta,t)$ in~\eqref{eq:Alg:h}, the equality in~\eqref{eq:theo:Convergence:05} can be expressed as
\begin{equation}\label{eq:theo:Convergence:06bb}
\lim_{t\rightarrow\infty}h(0,t)-h(\varepsilon(t),t)=0.
\end{equation}
Now, from~\eqref{eq:Alg:epsilon}, it is straightforward to show that
\begin{equation}\label{eq:theo:Convergence:06bc}
h(\varepsilon(t),t)-h(0,t)\leq\sigma\varepsilon(t)\frac{\partial h(\delta,t)}{\partial\delta}\bigg|_{\delta=0}\leq0,
\end{equation}
where the second inequality follows from~\eqref{eq:theo:Convergence:07} and~\eqref{eq:theo:Convergence:03b}. Computing the limit in~\eqref{eq:theo:Convergence:06bc} and using~\eqref{eq:theo:Convergence:06bb}, we obtain
\begin{equation}\label{eq:theo:Convergence:06bd}
\lim_{t\rightarrow\infty}\varepsilon(t)\frac{\partial h(\delta,t)}{\partial\delta}\bigg|_{\delta=0}=0.
\end{equation}
Hence, either $\lim_{t\rightarrow\infty}\frac{\partial h(\delta,t)}{\partial\delta}\big|_{\delta=0}=0$ or $\lim_{t\rightarrow\infty}\varepsilon(t)=0$. Our objective now is to show that, if $\lim_{t\rightarrow\infty}\varepsilon(t)=0$, then $\lim_{t\rightarrow\infty}\frac{\partial h(\delta,t)}{\partial\delta}\big|_{\delta=0}=0$. Towards that end, we note that, since the second order derivative of $f$ is bounded, we can find a constant $D>0$ such that, for all $\delta$ and $t$, $\frac{\partial^2h(\delta,t)}{\partial\delta^2}\leq D$. Integrating both sides of this inequality twice yields
\begin{equation}\label{eq:theo:Convergence:08}
h(\delta,t)\leq\frac{D}{2}\delta^2+\frac{\partial h(x,t)}{\partial x}\bigg|_{x=0}\delta+h(0,t).
\end{equation}
Since~\eqref{eq:theo:Convergence:08} is an upper bound on $h(\delta,t)$, we can use the Armijo stepsize rule to find a lower bound on $\varepsilon(t)$. In particular, choosing $\delta^*$ to be the largest value of $\{\chi^n\}_{n=0,1,\dots}$ satisfying
\begin{multline}\label{eq:theo:Convergence:08b}
\frac{D}{2}\delta^{*2}+\frac{\partial h(x,t)}{\partial x}\bigg|_{x=0}\delta^*+h(0,t)\leq\\
h(0,t)+\sigma\delta^*\frac{\partial h(x,t)}{\partial x}\bigg|_{x=0},
\end{multline}
then, from~\eqref{eq:theo:Convergence:08}, we have that
\begin{equation}\label{eq:theo:Convergence:08bb}
h(\delta^*,t)\leq h(0,t)+\sigma\delta^*\frac{\partial h(x,t)}{\partial x}\bigg|_{x=0},
\end{equation}
which, from the definition of $\varepsilon(t)$ in~\eqref{eq:Alg:epsilon}, implies that $\varepsilon(t)\geq\delta^*$. The inequality in~\eqref{eq:theo:Convergence:08b} yields
\begin{equation}\label{eq:theo:Convergence:08c}
\delta^*\leq\frac{2(\sigma-1)}{D}\frac{\partial h(x,t)}{\partial x}\bigg|_{x=0}.
\end{equation}
Hence, if $\lim_{t\rightarrow\infty}\varepsilon(t)=0$, the facts that $\varepsilon(t)\geq\delta^*$ and that $\delta^*$ is the largest value of $\{\chi^n\}_{n=0,1,\dots}$ satisfying~\eqref{eq:theo:Convergence:08c} imply that $\lim_{t\rightarrow\infty}\frac{\partial h(\delta,t)}{\partial\delta}\big|_{\delta=0}=0$, which completes the proof of the theorem.


\vspace*{-0.5\baselineskip}

\begin{IEEEbiographynophoto}{Daniel Calabuig} (M'11)
received the M.Sc. and Ph.D. degrees in telecommunications from the Universidad Polit\'ecnica de Valencia (UPV), Valencia, Spain, in 2005 and 2010 respectively.
In 2005 he joined the Institute of Telecommunications and Multimedia Applications (iTEAM) from the UPV. During his Ph.D. he participated in some European projects and activities like NEWCOM, COST2100 and ICARUS, working on radio resource management in heterogeneous wireless systems and Hopfield neural networks optimization. In 2009 he visited the Centre for Wireless Network Design at the University of Bedfordshire, Luton, UK. In 2010 he obtained a Marie Curie Fellowship from the European Commission to research on cooperative multipoint transmissions. Thanks to this fellowship, Daniel Calabuig visited the Department of Systems and Computer Engineering at Carleton University, Ottawa, Canada, from 2010 to 2012. During 2012, he also visited the TOBB Ekonomi ve Teknoloji \"Universitesi, Ankara, Turkey. In 2012, he returned to the iTEAM and started working inside the European project Mobile and wireless communications Enablers for the Twenty-twenty Information Society (METIS), which main objective is laying the foundation of 5G, the next generation mobile and wireless communications system. He is currently involved in the METIS-II project.
\end{IEEEbiographynophoto}

\vspace*{-0.5\baselineskip}

\begin{IEEEbiographynophoto}{Ramy H. Gohary} (S'02--M'06--SM'13)
is an adjunct research professor and a senior research associate with the Systems and Computer Engineering Department at Carleton University. He received  the B.Sc.\ (Hons.) degree  from Assiut University, Egypt in 1996, the M.Sc.\ degree from Cairo University,  Egypt, in 2000,  and the Ph.D. degree   from McMaster University, Ontario, Canada in 2006, all in electronics and communications engineering.

Dr. Gohary is the co-inventor of five US patents, and the author of more than thirty  well-cited IEEE journal papers. He is also the  referee for more than ten scientific IEEE journals, and a member of the technical program committees of seven international IEEE conferences.

Dr. Gohary is a registered Limited Engineering Licensee~(LEL) in the province of Ontario.

Dr. Gohary's research  interests include analysis and design of MIMO and cooperative wireless communication systems, applications of optimization and geometry  in signal processing and communications, information theoretic aspects of multiuser communication systems,  and applications  of iterative detection and decoding  techniques in multiple antenna  and multiuser systems.
\end{IEEEbiographynophoto}

\begin{IEEEbiographynophoto}{Halim Yanikomeroglu}
(S'96-M'98-SM'12) was born in Giresun, Turkey, in 1968. He received the B.Sc. degree in electrical and electronics engineering from the Middle East Technical University, Ankara, Turkey, in 1990, and the M.A.Sc. degree in electrical engineering (now ECE) and the Ph.D. degree in electrical and computer engineering from the University of Toronto, Canada, in 1992 and 1998, respectively.

During 1993-1994, he was with the R\&D Group of Marconi Kominikasyon A.S., Ankara, Turkey. Since 1998 he has been with the Department of Systems and Computer Engineering at Carleton University, Ottawa, Canada, where he is now a Full Professor. His research interests cover many aspects of wireless technologies with a special emphasis on wireless networks. In recent years, his research has been funded by Huawei, Telus, Blackberry, Samsung, Communications Research Centre of Canada (CRC), DragonWave, and Nortel. This collaborative research resulted in over 20 patents (granted and applied). Dr. Yanikomeroglu has been involved in the organization of the IEEE Wireless Communications and Networking Conference (WCNC) from its inception, including serving as Steering Committee Member as well as the Technical Program Chair or Co-Chair of WCNC 2004 (Atlanta), WCNC 2008 (Las Vegas), and WCNC 2014 (Istanbul). He was the General Co-Chair of the IEEE Vehicular Technology Conference Fall 2010 held in Ottawa. He has served in the editorial boards of the IEEE TRANSACTIONS ON COMMUNICATIONS, IEEE TRANSACTIONS ON WIRELESS COMMUNICATIONS, and IEEE COMMUNICATIONS SURVEYS \& TUTORIALS. He was the Chair of the IEEEs Technical Committee on Personal Communications (now called Wireless Technical Committee). He is a Distinguished Lecturer for the IEEE Communications Society and the IEEE Vehicular Technology Society.

Dr. Yanikomeroglu is a recipient of the IEEE Ottawa Section Outstanding Educator Award in 2014, Carleton University Faculty Graduate Mentoring Award in 2010, the Carleton University Graduate Students Association Excellence Award in Graduate Teaching in 2010, and the Carleton University Research Achievement Award in 2009. Dr.
Yanikomeroglu spent the 2011-2012 academic year at TOBB University of Economics and Technology, Ankara, Turkey, as a Visiting Professor. He is a registered Professional Engineer in the province of Ontario, Canada.
\end{IEEEbiographynophoto}
\end{document}